\documentclass[preprintnumbers,aps,amsmath,amssymb,pra,twocolumn,showpacs,superscriptaddress,floatfix]{revtex4-1}
\usepackage[dvipdfmx]{graphicx}
\usepackage{amsmath, amssymb, amsthm}
\usepackage{mathtools}
\usepackage{algorithm, algpseudocode}
\usepackage{color}
\usepackage{float}

\def\Bra#1{\left\langle#1\right|}
\def\Ket#1{\left|#1\right\rangle}

\newcommand{\ket}[1]{|{#1}\rangle}

\DeclareMathOperator{\rank}{rank}
\DeclareMathOperator{\tr}{Tr}

\theoremstyle{plain}
\newtheorem{theorem}{Theorem}
\newtheorem{lemma}[theorem]{Lemma}
\newtheorem{corollary}[theorem]{Corollary}
\newtheorem{proposition}[theorem]{Proposition}
\theoremstyle{definition}
\newtheorem{definition}[theorem]{Definition}
\newtheorem{example}{Example}
\theoremstyle{remark}
\newtheorem{conjecture}{Conjecture}
\newtheorem{remark}[conjecture]{Remark}

\begin{document}

\title{Graph-associated entanglement cost of a multipartite state in exact and finite-block-length approximate constructions}

\author{Hayata Yamasaki}
\email{yamasaki@eve.phys.s.u-tokyo.ac.jp}
\author{Akihito Soeda}
\email{soeda@phys.s.u-tokyo.ac.jp}
\author{Mio Murao}
\email{murao@phys.s.u-tokyo.ac.jp}
\affiliation{Department of Physics, Graduate School of Science, The University of Tokyo, Tokyo, Japan}

\date{\today}

\begin{abstract}
    We introduce and analyze \textit{graph-associated entanglement cost}, a generalization of the entanglement cost of quantum states to multipartite settings.  We identify a necessary and sufficient condition for any multipartite entangled state to be constructible when quantum communication between the multiple parties is restricted to a  quantum network represented by a tree.  The condition for exact state construction is expressed in terms of the Schmidt ranks of the state defined with respect to edges of the tree.  We also study approximate state construction and provide a second-order asymptotic analysis.
\end{abstract}

\pacs{03.67.Mn, 03.67.Ac, 03.67.Bg}
\keywords{multipartite entanglement, entanglement cost, Schmidt rank, network, tree (graph theory), finite block length}

\maketitle

\section{\label{sec:1}Introduction}
Multipartite entanglement~\cite{rev1,rev2,RefWorks:111,rev3,rev4} serves as resource in multi-party quantum tasks, when the parties are restricted to local operations and classical communication (LOCC).  These nonlocal tasks include quantum cryptographic protocols~\cite{secret_sharing,8} and quantum computation~\cite{RefWorks:164}.  One way to characterize multipartite entanglement is to analyze the amount of quantum communication required for constructing the corresponding entangled state from a separable state.  Under LOCC, a noiseless quantum channel and bipartite maximally entangled state are equivalent by means of quantum teleportation~\cite{RefWorks:21}.
For bipartite states, there are two types of scenarios to consider when evaluating the amount of entanglement of a given state in terms of the consumed resource.  The first type is an asymptotic scenario, where the construction assumes an infinitely increasing number of copies of the given state.  The cost in this case is the minimum asymptotic rate of quantum communication per copy to achieve the construction and referred as the entanglement cost~\cite{RefWorks:139, RefWorks:131}.   While this asymptotic cost may be understood as a fundamental quantity in quantum information theory, the other type based on the communication cost in one-shot scenario is more relevant when the aim is to generate a finite number of copies of the target entangled state.  One-shot costs may be further distinguished by whether the state construction is required to be exact~\cite{RefWorks:157} or certain error in fidelity is tolerated~\cite{RefWorks:158,RefWorks:160}.

\begin{figure}
\centering
\includegraphics[width=8cm]{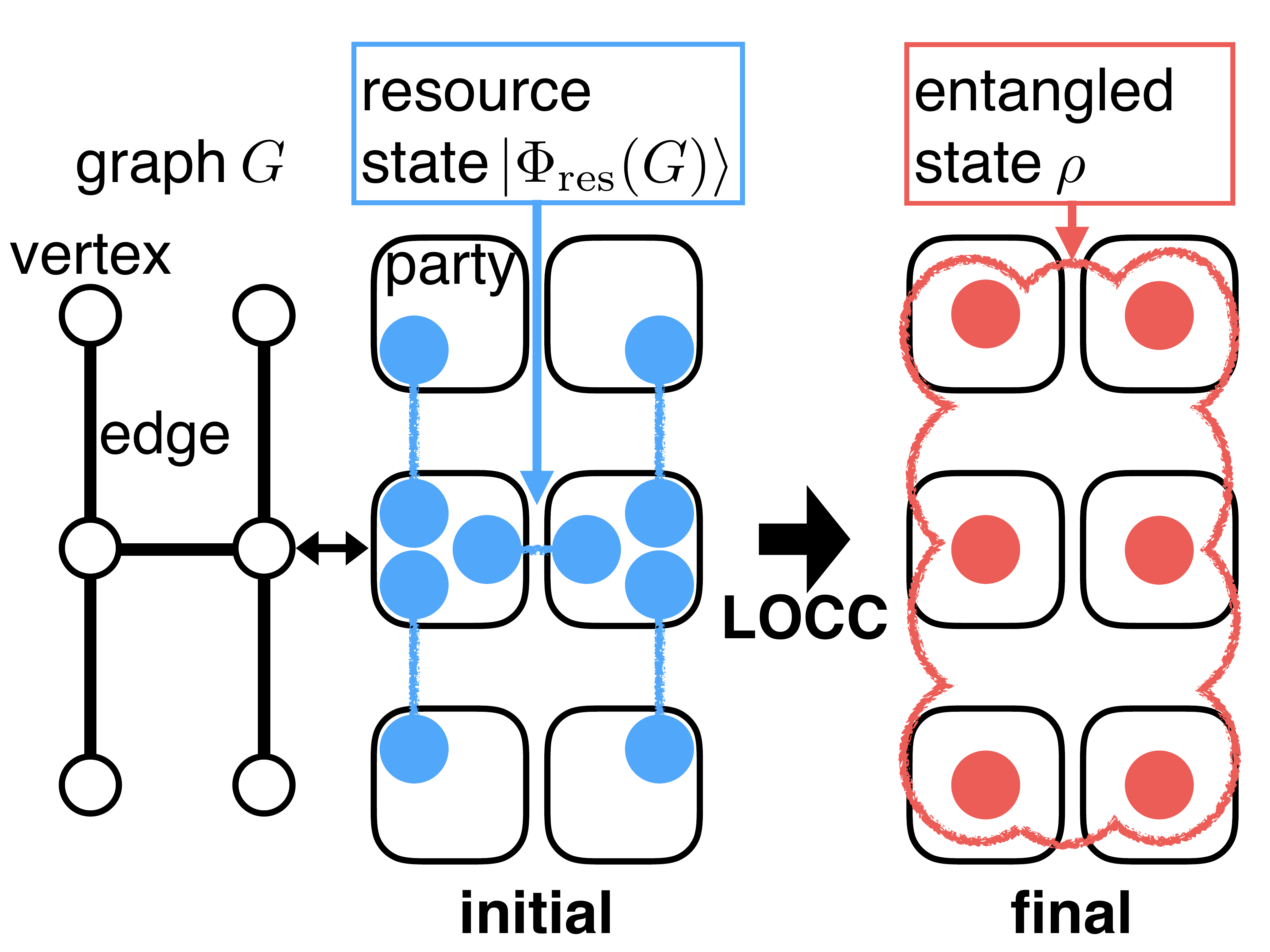}
\caption{(Color online) Schematic representation of our LOCC construction of a multipartite entangled state under graph $G$.  Each solid circle represents a quantum system, possibly of a different dimension, while each rounded square corresponds to a party, which may possess several quantum systems within.  The parties are connected by noiseless quantum communication channels specified by the edges of $G$.  Under LOCC, each use of a channel is equivalent to consuming a maximally entangled state, given by a pair of blue solid circles with a line in between (\textit{initial}).  The target state $\rho$ is represented by red solid circles  encircled in a bubble (\textit{final}).  The maximally entangled states all together comprise a resource state $\Ket{\Phi_{\textup{res}}(G)}$ for LOCC construction of $\rho$.
}
\label{fig:intro}
\end{figure}

In this paper, we introduce a characterization of multipartite entangled states by generalizing the cost-based measure of bipartite entangled states for their LOCC construction.  The multiple parties in our case are connected by several noiseless quantum communication channels, whose 
 connectivity is represented by a graph of the graph theory~\cite{RefWorks:152}, described in terms of vertices and edges.
Although the total number of edges does not fully characterize a graph, any connected network of $N$ parties requires at least $N-1$ channels.   If a connected $N$-vertex graph has exactly $N-1$ edges, then it is called a \textit{tree}.  As in the bipartite case, a single use of each channel between two parties is equivalent to a maximally entangled state of fixed dimensional systems consumed under LOCC\@.  Thus the set of the maximally entangled states distributed according to the edges of the graph as illustrated in Fig.~\ref{fig:intro} characterizes the nonlocal resources for constructing a multipartite state.
The total amount of quantum communication is evaluated by the total amount of the maximally entangled states and generalizes the notion of entanglement cost, which we name \textit{graph-associated entanglement cost}.
Our choice of resource states differs from Ref.~\cite{32}, which adopts Greenberger-Horne-Zeilinger (GHZ) states for the resource.   Reference~\cite{45} considers more general quantum channels with noise, but restricts its analysis to two-colorable graph states.

We analyze the graph-associated entanglement cost of multipartite \textit{pure} states on tree networks to achieve exact and approximate state construction.  We prove that the former is given in terms of the Schmidt rank~\cite{RefWorks:157} defined with respect to edges of the tree.  For the latter, we refine the analysis sketched in Ref.~\cite{RefWorks:148} and combine the results of Ref.~\cite{RefWorks:160} to provide the second-order asymptotic analysis, utilizing the exact state construction algorithm.  We observe that the multipartite structure appears not in the first-order but in the second-order coefficients of the quantum communication rate.  Our exact state construction algorithm generalizes the one presented in Refs.~\cite{PRL110503,PRA032311}, which in our terminology assumes a network of a straight line.

The rest of the paper is organized as follows.  In Section~\ref{sec:2}, we define the exact and approximate state construction tasks, and formalize the graph-associated entanglement costs corresponding to the state construction tasks.  We analyze state construction under trees introduced in Section~\ref{sec:tree}.  The exact state construction is analyzed in Section~\ref{sec:3} where an implication for construction of projected entangled pair states (PEPS)~\cite{PEPS} is also presented.  The approximate state construction is analyzed in Section~\ref{sec:4}.  Our conclusion is given in Section~\ref{sec:5}.

\section{\label{sec:2} Graph-associated entanglement cost for multipartite state construction}
\subsection{Definition of construction tasks}
We introduce tasks of \textit{exact construction} and \textit{approximate construction} of a multipartite entangled state shared among $N$ parties from a fixed product state of $N$ parties when quantum communication between parties is restricted to specified parties.   When the direction of quantum communication is not restricted, we use a simple undirected graph $G=(V(G),E(G))$ to represent the restriction on quantum communication.  Each of $N$ vertices $v\in V(G)=\{v_1,v_2,\ldots,v_N\}$ of $G$ represents one of the $N$ parties, and quantum communication is only allowed between the parties $v_i$ and $v_j$ directly connected by an edge $e = \left\{v_i,v_j\right\} \in E(G)$ of $G$, where we identify $\{v_i, v_j\}$ with $\{v_j, v_i\}$.   To construct an arbitrary entangled state shared between $N$ parties,  $G$ is required to be a connected graph. We omit the arguments of $V(G)$ and $E(G)$ to simply write $G=(V,E)$ if obvious.

Sending a quantum state in an $m_e$-dimensional Hilbert space between two parties $v_i$ and $v_j$ represented by edge $e = \left\{v_i,v_j\right\} $ can be achieved by quantum teleportation using, in addition to local operations by each party and classical communication between the parties, a maximally entangled state of two $m_e$-level systems shared between $v_i$ and $v_j$ defined by 
\begin{align*}
    \Ket{\Phi_{m_e}^+}^e\coloneqq\frac{1}{\sqrt{m_e}}\sum_{l=0}^{m_e-1}\Ket{l}^{v_i}\otimes\Ket{l}^{v_j},
\end{align*}
where the superscripts for each ket represent parties to which the state belongs.
To represent parties with the superscripts of kets, we may write numbers representing the vertices for brevity, such as $\Ket{\Phi_{m_e}^+}^{i,j}$ on $v_i$ and $v_j$.
If we do not impose restrictions on the classical communication and the local operations apart from the laws of quantum mechanics, the maximally entangled state is a \textit{resource state} for performing quantum communication by local operations and classical communication (LOCC), while LOCC is considered to be nonrestricted ``free'' resources in this setting.

In the LOCC framework, the tasks of constructing an $N$-partite entangled state under the restriction on quantum communication is equivalent to the tasks of transforming an initial state consisting of a set of bipartite resource states specified by a set of edges $E$ into the target multipartite entangled state to be constructed under LOCC\@.    The initial resource state for a given graph $G$, which is initially shared by the $N$ parties, is represented by 
\begin{align*}
\Ket{\Phi_{\textup{res}} (G)} \coloneqq \bigotimes_{e\in E}\Ket{\Phi_{m_e}^+}^e,
\end{align*} 
where $m_e$ is the dimension of the initial resource state specified by edge $e$.
We write the Hilbert space for the initial resource state $\Ket{\Phi_{\textup{res}} (G)}$ as 
\[
  \mathcal{H}_r \coloneqq \bigotimes_{v\in V}\mathcal{H}_{r}^v,
\]
where, for each $v\in V$, $\mathcal{H}_{r}^v$ denotes the total Hilbert space of the initial resource state belonging to party $v$.
For each edge $e=\{v,v'\}$, the $m_e$-dimensional subsystem of $\mathcal{H}_{r}^v$ for $v$'s part of $\Ket{\Phi_{m_e}^+}^e$ is denoted by $\mathcal{H}_{r_{e}}^v$, or $\mathcal{H}_{r_{\{v,v'\}}}^v$ equivalently, and we write the Hilbert space for the initial resource state shared between $v$ and $v'$ as	
\[
  \mathcal{H}_r^e \coloneqq \mathcal{H}^{v}_{r_e}\otimes\mathcal{H}^{v'}_{r_e} \ni \Ket{\Phi_{m_e}^+}^e,
\]
or $\mathcal{H}_r^{\{v,v'\}}$ equivalently.

We denote an $N$-partite target state, which is to be constructed and shared by the $N$ parties, by a density operator $\rho$, and the Hilbert space for $\rho$ by
\[
  \mathcal{H}_t \coloneqq \bigotimes_{v\in V} \mathcal{H}_t^v,
\]
where $\mathcal{H}_t^v$ represents the $d_v$-dimensional Hilbert space of the target state belonging to party $v$.

A construction task under a graph $G$ for a given target state $\rho$ is to deterministically transform the initial resource state $\Ket{\Phi_{\textup{res}} (G)}$ on $\mathcal{H}_r$ into a target state $\rho$ on $\mathcal{H}_t$ by LOCC\@.
Note that, in the LOCC framework, each party $v$ can add an auxiliary system of arbitrary dimension, represented by a Hilbert space $\mathcal{H}^v_a$, and can perform arbitrary operations on $\mathcal{H}_r^v \otimes \mathcal{H}_a^v$, which can be conditioned by other parties' measurement outcomes obtained by classical communication.

Formally, the exact and approximate state construction tasks are defined in the following.
\begin{definition}[Exact state construction]
The \textit{exact state construction task} under graph $G=(V,E)$ for an $N$-partite target state $\rho$ on $\mathcal{H}_t$ is to deterministically and exactly transform the initial resource state $\Ket{\Phi_{\textup{res}} (G)}\in \mathcal{H}_{r}$ into $\rho$ on $\mathcal{H}_t$ by an LOCC operation $\Gamma_{\textup{LOCC}}$, namely,
\[
\rho=\Gamma_\textup{LOCC}\left(\Ket{\Phi_\textup{res}(G)}\right),
\]
where each of the $N$ vertices of $G$ corresponds to a party in $\Gamma_\textup{LOCC}$.
\end{definition}
\begin{definition}[$(n, \epsilon)$-approximate state construction]
The \textit{$(n, \epsilon)$-approximate state construction task} under graph $G=(V,E)$ for an $N$-partite target state $\rho$ on $\mathcal{H}_t$ is to deterministically transform the initial resource state $\Ket{\Phi_{\textup{res}} (G)}$ to a state $\tilde \rho_n$ which approximates $n$ copies of the target state $\rho^{\otimes n}$ within a preset error threshold $\epsilon >0$ in terms of the trace distance $\left\|\tilde\rho_n - \rho^{\otimes n}\right\| _1 \leq \epsilon$ by an LOCC operation $\Gamma_{\textup{LOCC}}$, namely,
\begin{align*}
\tilde{\rho}_n=\Gamma_{\textup{LOCC}} (\Ket{\Phi_{\textup{res}} (G)} ).
\end{align*}
\end{definition}

Note that, in our definition, for each party $v$, the dimension of $\mathcal{H}^v_r$ for the initial resource state depends on the number of the neighboring parties, and can be different from that of $\mathcal{H}^v_t$ for the target state.  This setting is different from other LOCC convertibility analyses where states are  from the same Hilbert space~\cite{RefWorks:161, PhysRevA.83.022331, arxiv:1606.04418, arxiv:1607.05145}.

For the case of $N=2$, the exact state construction is possible if and only if the Schmidt number of the target state, a generalization of the Schmidt rank of bipartite pure states to mixed states, is equal to or smaller than that of the initial resource state~\cite{RefWorks:157}.    Also, for $N=2$, the approximate state construction task reduces to entanglement dilution in the finite block length regime~\cite{RefWorks:158, RefWorks:159, HierarchyOfInformationQuantities, RefWorks:160}, which is used for defining the entanglement cost.

\subsection{Total graph-associated entanglement cost}

We first define \textit{total} graph-associated entanglement cost based on the exact and approximate construction of multipartite states shared among $N$ parties under restricted quantum communication specified by graph $G=(V,E)$.   Similar to the case of bipartite states, exact (or approximate) total entanglement costs of a state $\rho$ are defined as the ``minimum required amount of entanglement'' of the initial state $\Ket{\Phi_{\textup{res}} (G)}$ so to be exactly (or approximately) transformable to the target state $\rho$ by LOCC\@.  

Since the initial resource state $\Ket{\Phi_{\textup{res}} (G)}$ only consists of a set of bipartite resource states ${\left\{ \Ket{\Phi_{m_e}^+}^e \right\}}_{e\in E}$, we define the total amount of entanglement of  $\Ket{\Phi_{\textup{res}} (G)}$ to be the sum of the amount of bipartite entanglement of each $\Ket{\Phi_{m_e}^+}^e$ and denote the sum by $E_\textup{sum} (\Ket{\Phi_{\textup{res}} (G)})$.   By using the unit of \textit{ebit}, which is the entanglement entropy of the \textit{Bell state} $\Ket{\Phi_{2}^+}$ (a maximally entangled two-qubit state), the amount of entanglement of $\Ket{\Phi_{m_e}^+}^e$ is given by $\log_2 m_e$.    Therefore the total amount of entanglement of $\Ket{\Phi_{\textup{res}} (G)}$ is given by $E_\textup{sum} (\Ket{\Phi_{\textup{res}} (G)}) =\sum_{e \in E} \log_2 m_e$.

For a quantum state $\rho$ on $\mathcal{H}_t$ and a graph $G=(V,E)$, the exact total graph-associated entanglement cost of $\rho$ is defined by
\begin{align*}
E_{\textup{GC}}^G (\rho) \coloneqq \min_{\Ket{\Phi_{\textup{res}} (G)} \in \mathcal{R}_G (\rho)} E_\textup{sum} (\Ket{\Phi_{\textup{res}} (G)}),
\end{align*}
where $\mathcal{R}_G (\rho)$ represents a set of states $\{ \Ket{\Phi_{\textup{res}} (G)} \}$ such that the exact state construction task under $G$ for $\rho$ is achievable. 

Similarly, for a quantum state $\rho$ on $\mathcal{H}_t$ and a graph $G=(V,E)$, the $(n, \epsilon)$-approximate total graph-associated entanglement cost is defined by 
\begin{align*}
E_{\textup{GC}}^{n, \epsilon,G} (\rho) \coloneqq \frac{1}{n} \min_{\Ket{\Phi_{\textup{res}} (G)} \in \mathcal{R}^{n,\epsilon}_G (\rho)} E_\textup{sum} (\Ket{\Phi_{\textup{res}} (G)}),
\end{align*}
where $\mathcal{R}^{\epsilon, n}_G (\rho)$ represents a set of states $\{ \Ket{\Phi_{\textup{res}} (G)} \}$ such that the $(n, \epsilon)$-approximate state construction task under $G$ for $\rho$ is achievable. 

\subsection{Edge graph-associated entanglement cost}

We define the exact and approximate \textit{edge} graph-associated entanglement cost to characterize distributed entanglement properties of multipartite states.  We define the set of optimal initial resource states $\hat{\mathcal{R}}_G (\rho) \subset \mathcal{R}_G (\rho)$ minimizing the exact total graph-associated entanglement cost $E_{\textup{GC}}^G (\rho)$ and assign an index $i$ to represent different configurations of optimal bipartite resource states.   The elements of $\hat{\mathcal{R}}_G (\rho)$ are denoted by
\begin{align}
    \hat{\mathcal{R}}_G (\rho) = {\left\{ \Ket{\hat{\Phi}_{\textup{res}}^i (G,\rho) } \coloneqq \bigotimes_{e \in E} \Ket{\Phi_{m_e(i)}^+}^e\right\}}_i.  
\label{Rdef}
\end{align}
Using $m_e(i)$ defined by Eq.~\eqref{Rdef}, the exact edge graph-associated entanglement cost is defined for each configuration $i$ as follows.
\begin{definition}[Exact edge graph-associated entanglement cost]
    For a graph $G=(V,E)$ and an $N$-partite quantum state $\rho$ on $\mathcal{H}_t$, the \textit{exact edge graph-associated entanglement cost} of $\rho$ with respect to edge $e\in E$ and for a configuration $i$ is defined by
\begin{align*}
E_{\textup{GC},i,e}^G (\rho) \coloneqq \log_2 m_e (i),
\end{align*}
where $m_e (i)$ represents the dimension of the bipartite resource state corresponding to edge $e$ of the optimal initial resource states for the exact state construction task under graph $G$ for $\rho$ with configuration $i$ defined by Eq.~\eqref{Rdef}.
\end{definition}

Similarly, we define a set of optimal initial resource states $\hat{\mathcal{R}}_G^{n, \epsilon} (\rho)$ minimizing the $(n,\epsilon)$-approximate total graph-associated entanglement cost $E_{\textup{GC}}^{n, \epsilon, G} (\rho)$ and we assign an index $i$ to represent different configurations.  The elements of $\hat{\mathcal{R}}_G^{n, \epsilon} (\rho) $ are denoted by 
\begin{align}
    \hat{\mathcal{R}}_G^{n, \epsilon} (\rho) = {\left\{ \Ket{\hat{\Phi}_{\textup{res}}^{i,n,\epsilon} (G,\rho) } \coloneqq \bigotimes_{e \in E} \Ket{\Phi_{m_e(i,n,\epsilon)}^+}^e\right\}}_i.  
\label{Rdef2}
\end{align}
Using $m_e(i,n,\epsilon)$ defined by Eq.~\eqref{Rdef2}, the $(n,\epsilon)$-approximate edge graph-associated entanglement cost is defined for each configuration $i$ as follows.
\begin{definition}[$(n,\epsilon)$-approximate edge graph-associated entanglement cost]
    For a graph $G=(V,E)$ and an $N$-partite quantum state $\rho$ on $\mathcal{H}_t$, the \textit{$(n,\epsilon)$-approximate edge graph-associated entanglement cost} of $\rho$ with respect to edge $e\in E$ and for a configuration $i$  is defined by 
\begin{align*}
E_{\textup{GC},i,e}^{n,\epsilon, G} (\rho) = \frac{1}{n} \log_2 m_e (i,n,\epsilon),
\end{align*}
where $m_e (i,n,\epsilon)$ represents the dimension of the bipartite resource state corresponding to edge $e$ of the optimal initial resource states for the $(n,\epsilon)$-approximate state construction task under graph $G$ for $\rho$ with configuration $i$ defined by Eq.~\eqref{Rdef2}.
\end{definition}

The asymptotic edge graph-associated entanglement cost with respect to edge $e$ is denoted by
\begin{align*}
E_{\textup{GC},i,e}^{\infty, G} (\rho) \coloneqq  
\lim_{\epsilon \rightarrow 0} \lim_{n \rightarrow \infty} \frac{1}{n} \log_2 m_e (i,n,\epsilon).
\end{align*}

By calculating the edge graph-associated entanglement costs for all the edges of a graph $G$,
we can derive the initial resource state from which construction of multipartite states under graph $G$ is achievable with the smallest amount of total entanglement.
By definition, $E_{\textup{GC},i,e}^G (\rho)$ and $E_{\textup{GC},i,e}^{n,\epsilon, G} (\rho)$ satisfy $E_{\textup{GC}}^G (\rho)=\sum_{e \in E} E_{\textup{GC},i,e}^G (\rho)$ and $E_{\textup{GC}}^{n,\epsilon, G} (\rho)  = \sum_{e \in E} E_{\textup{GC},i,e}^{n,\epsilon, G} (\rho)$.

\section{\label{sec:tree}Analysis under trees}
\subsection{Tree resource states}
Calculating the edge graph-associated entanglement costs for an arbitrary graph is difficult since their definitions include optimization of the total graph-associated entanglement costs.   Thus we focus on analyzing construction tasks under a special class of graphs, \textit{trees}, for the initial resource states.   A network represented by a tree describes the situation where all parties are connected by the smallest number of channels.  

Trees are connected graphs which contain no cycle.
Trees with $N$ vertices have $N-1$ edges, which is the minimum to connect all the vertices.
Any connected graphs can be reduced to a tree spanning all the vertices by removing some of the edges.
We let $T=(V,E)$ denote a tree. 
For a tree $T=(V,E)$,  we can designate any vertex as the root of the tree $T$, which we label as $v_1\in V$.
We represent such a rooted tree as $(T,v_1)$.
Without loss of generality, we label the vertices and edges of $T$ so that the other vertices $v_2, v_3,\ldots, v_N$ and the edges $e_1,e_2,\ldots,e_{N-1}$ are in order of a breadth-first search starting from $v_1$~\cite{RefWorks:142}.  In other words, the closer to the root a vertex (or an edge) is, the smaller the number of the label of the vertex (or the edge) is.
We also let $p(v)$ denote $v$'s parent. The \textit{tree resource state} is denoted by $\Ket{\Phi_{\textup{res}}(T)}$.

\begin{figure}
\centering
\includegraphics[width=8cm]{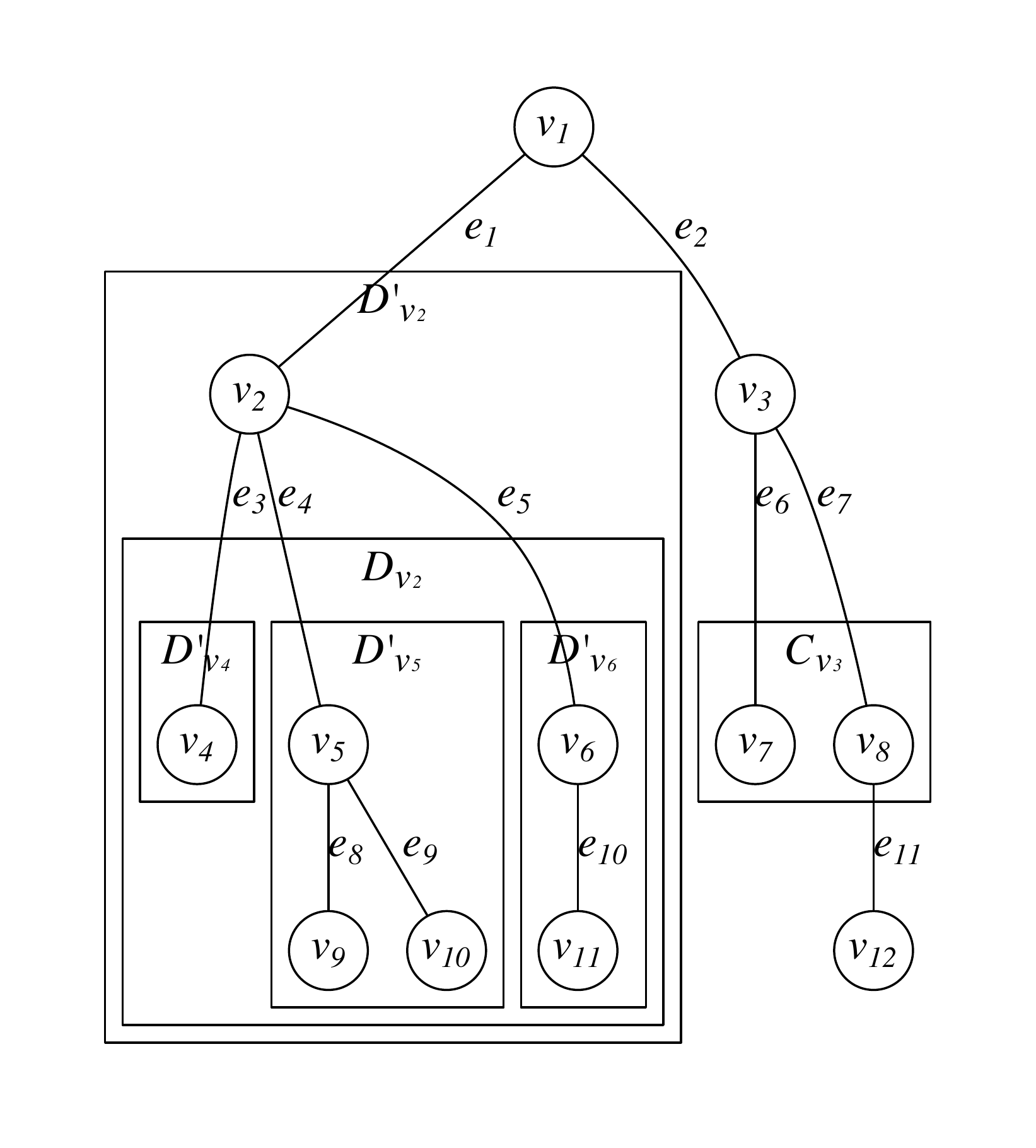}
\caption{A graphical representation of a tree.
    The vertices and edges are labeled where $v_1$ is the root of the tree and the other vertices $v_2, v_3,\ldots, v_N$ and the edges $e_1,e_2,\ldots,e_{N-1}$ are in order of a breadth-first search starting from $v_1$.   
  $C_v$, $D_{v}$ and $D'_v$ denote the set of $v$'s children, the set of $v$'s descendants, and the set of $v$ itself and $v$'s descendants, respectively. $D'_v$ is decomposed into $v$ itself and $D'_c$'s for all $c\in C_v$.
}
\label{fig:tree}
\end{figure}

We introduce a recursive structure according to rooted trees, as illustrated in Fig.~\ref{fig:tree}.
For any $v\in V$ of the rooted tree $(T,v_1)$ where $T=(V,E)$, we let $C_v$ denote the set of $v$'s children, $D_{v}$ the set of $v$'s descendants, and $D'_v$ the set of $v$'s descendants and $v$ itself.
We refer to a vertex without any child as a leaf of the rooted tree.
For any $v\in V$, $D'_v$ can be decomposed by using these notations as
\begin{equation}
    \label{eq:v_each}
    \begin{split}
        D'_v &= \{v\}\cup D_v\\
        &=\{v\}\cup \bigcup_{c\in C_v} D'_c.
    \end{split}
\end{equation}
The set of all vertices $V$ is represented by $V=D'_{v_1}$ for the root specified by $v_1$.  Using Eq.~\eqref{eq:v_each}, we can recursively decompose $V$ according to the given rooted tree.

The Hilbert space $\mathcal{H}_t$ for the target state can also be recursively decomposed.   By denoting the Hilbert space corresponding to the set of vertices $D'_v$ as $\mathcal{H}_t^{D'_v}$,  it is decomposed as 
\begin{equation}
    \label{eq:h_each}
    \begin{split}
        \mathcal{H}_t^{D'_v} &= \mathcal{H}_t^v\otimes\mathcal{H}_t^{D_v}\\
        &=\mathcal{H}_t^{v}\otimes\bigotimes_{c\in C_v}\mathcal{H}_t^{D'_c},
    \end{split}
\end{equation}
for any nonleaf $v$.   Note that, by definition, $\mathcal{H}_t^{D'_{v_1}}=\mathcal{H}_t$.  We make use of the recursive decomposition given by Eq.~\eqref{eq:h_each} of $\mathcal{H}_t$ in our analysis of state construction tasks under trees.

\subsection{Construction of pure states under trees}
We focus on analyzing construction tasks for pure target states, while an extension to mixed states is possible by considering their purification or convex roof~\cite{RefWorks:111}.
A pure target state is denoted by $\Ket{\psi}$, where the corresponding density operator is written by $\psi \coloneqq \Ket{\psi}\Bra{\psi}$.

\begin{figure}
\centering
\includegraphics[width=8cm]{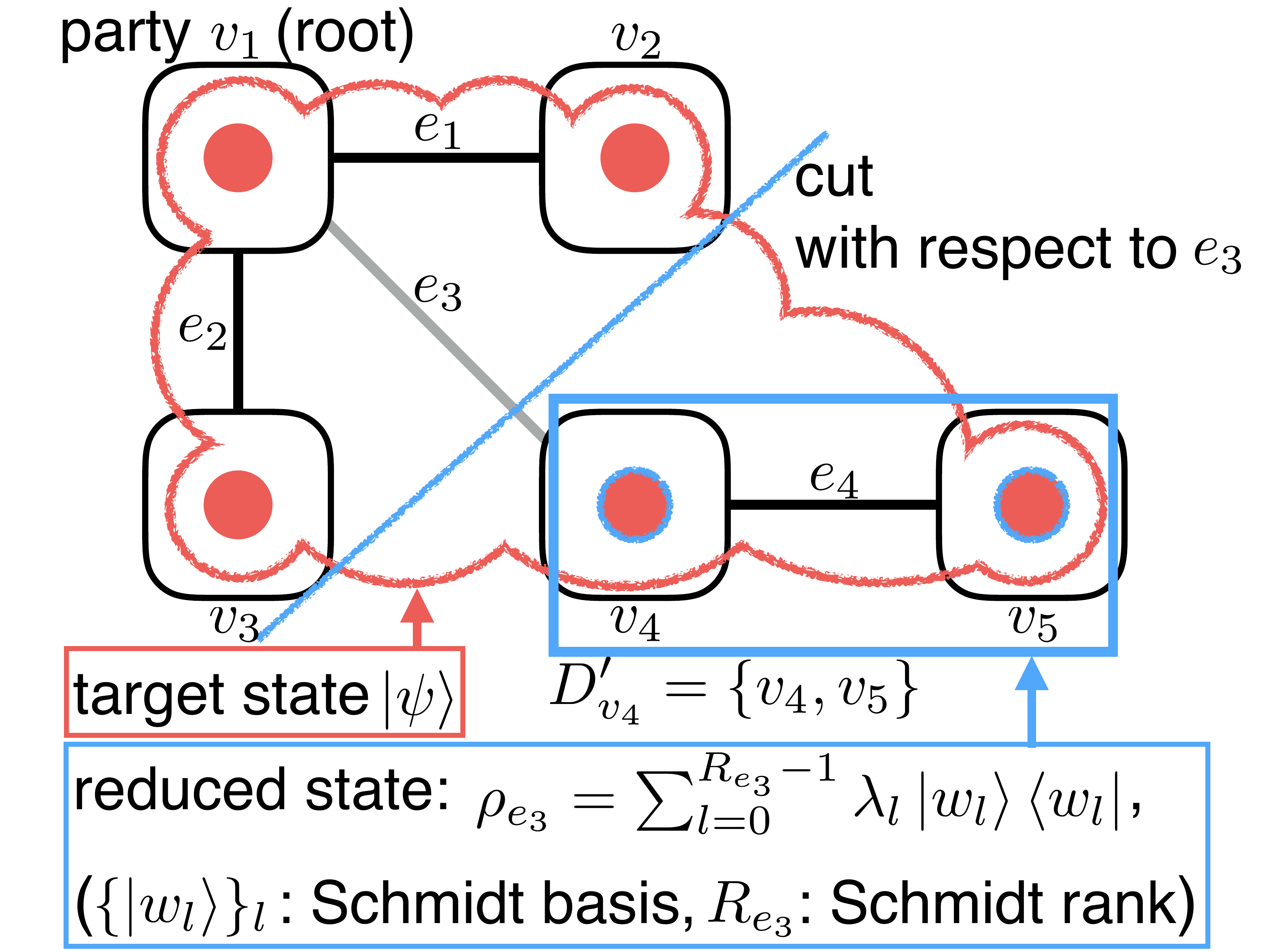}
\caption{(Color online) Schematic representation of a reduced density operator $\rho_e$ of the target state $\Ket\psi$ with respect to an edge $e$ of the rooted tree.
    For each edge $e=\{p(v),v\}$, the tree can be divided into two connected subgraphs corresponding to $D'_v$ and its complement.  We write the reduced density operator of $\Ket\psi$ corresponding to $D_v$ as $\rho_e$.  The Schmidt basis of $\mathcal{H}^{D'_v}_t$ for $\Ket\psi$ is written as $\left\{\Ket{w_l}\right\} _{l}$.  We can represent $\rho_e$ using a subset of $\left\{\Ket{w_l}\right\} _{l}$ corresponding to nonzero Schmidt coefficients.
}
\label{fig:density}
\end{figure}

To analyze edge graph-associated entanglement costs under trees, we define a reduced density operator of $\Ket{\psi}$ with respect to an edge $e\in E$, which we illustrate in Fig.~\ref{fig:density}.
Given a rooted tree $(T,v_1)$ where $T=(V,E)$, when any edge $e=\{p(v),v\}\in E$ is removed, the tree $T$ is divided into two connected subgraphs.   The set of vertices of one of the subgraphs including $v$ is $D'_v$.  We denote the set of vertices of the other subgraph including $p(v)$ by $\overline{D'_v}$, the complement of $D'_v$.    The bipartition of $T$ with respect to $e$ is referred to as the decomposition of the vertex set $V$ into $D'_v$ and $\overline{D'_v}$ induced by removing edge $e$.  A reduced state $\rho_e^{(T,v_1),\psi}$ of $\Ket{\psi}$ with respect to the edge $e=\{p(v),v\}$ is defined as the state on $\mathcal{H}_t^{D'_v}$ obtained by tracing out the systems belonging to $\overline{D'_v}$, namely,
\[
    \rho_e^{(T,v_1),\psi}\coloneqq \tr_{\overline{D'_v}}\psi,
\]
where $\tr_{\overline{D'_v}}$ denotes a partial trace of $\mathcal{H}_t^{\overline{D'_v}}$.
We omit the superscript $(T,v_1),\psi$ and simply write $\rho_e^{(T,v_1),\psi}$ as $\rho_e$ if obvious.   For simplicity, we also define  
\[
    R_e \coloneqq \rank \rho_{e}.
\]

We represent $\rho_e$ in terms of the Schmidt basis of $\Ket\psi$ for the bipartition with respect to $e$.
The Schmidt decomposition of $\Ket\psi$ for the bipartition with respect to $e=\{p(v),v\}$ is written as
\begin{equation}
    \label{eq:schmidt_decomposition}
    \Ket\psi = \sum_{l=0}^{R_e-1}\sqrt{\lambda_{l}}\Ket{w_{l}}\otimes\Ket{\overline{w}_{l}},
\end{equation}
where $\sqrt{\lambda_{l}}$ for each $l$ is a Schmidt coefficient satisfying $\sum_l \lambda_l = 1$, ${\left\{\Ket{w_l}\right\}}_{l}$ and ${\left\{\Ket{\overline{w}_{l}}\right\}}_{l}$ are the Schmidt bases of $\mathcal{H}_t^{D'_v}$ and $\mathcal{H}_t^{\overline{D'_v}}$ respectively.   In the following, we only consider an $R_e$-dimensional subspace spanned by the Schmidt basis states corresponding to nonzero Schmidt coefficients, and simply call the subset of the Schmidt basis states as the Schmidt basis for the bipartition with respect to $e$.  For any edge $e=\{p(v),v\}\in E$, the reduced density operator with respect to $e$ is given by
\[
    \rho_e=\sum_{l=0}^{R_e -1}\lambda_{l} \Ket{w_{l}}\Bra{w_{l}},
\]
where $\lambda_{l}> 0$ for each $l$ and $\sum_{l}\lambda_{l}=1$.

\section{\label{sec:3}Exact state construction under trees}

\subsection{Graph-associated entanglement cost for exact state construction}

For exact construction of bipartite pure states,  the initial resource state  $\Ket{\Phi_R^+}$ shared between two parties A and B can be transformed into any bipartite states having the Schmidt rank no more than $R$ by LOCC\@~\cite{RefWorks:157}.    As for exact construction of multipartite states, we obtain a similar result for tree resource states as follows.
\begin{theorem}[Exact edge graph-associated entanglement cost]
\label{thm:exact}
Given any tree $T=(V,E)$ and any $N$-partite pure target state $\Ket{\psi} \in \mathcal{H}_t$, the exact edge graph-associated entanglement cost for edge $e \in E$ under tree $T$ is given by
\[
    E_{\textup{GC},i,e}^T\left(\psi\right)=\log_2 {R_e}
\]
where $R_e = \rank\rho_e$  and the configuration $i$ of the optimal resource state is uniquely determined.
\end{theorem}

The lower bound of $E_{\textup{GC},i,e}^T\left(\psi\right)$ immediately follows from a bipartite entanglement property of $\ket{\psi}$ in terms of the bipartite cut induced by removing $e$ characterized by the rank of $\rho_e$, which is monotonically nonincreasing under LOCC\@.   In contrast, the upper bound is nontrivial, as the upper bound is not deduced from the bipartite case.   To show the upper bound coincides with the lower bound,  we present a distributed algorithm for exact state construction in which, for any tree $T=(V,E)$, any target state $\Ket{\psi}$ is exactly constructed from the tree resource state $\ket{\Phi_{\textup{res}}(T)}=\bigotimes_{e\in E}\Ket{\Phi^+_{R_e}}^e$ with the smallest amount of entanglement at each edge which allows the construction.  The state construction algorithm saves the maximum quantum memory space required for parties, as each party directly transforms the optimal initial resource state into its reduced target state by performing measurements, which can be performed with one auxiliary qubit at each party by employing a method introduced in Ref.~\cite{PhysRevA.77.052104}.

The state construction algorithm is based on a recursive description of $\Ket{\psi}$ according to a tree.  We recursively represent $\Ket{\psi} \in \mathcal{H}_t$ in terms of the Schmidt basis ${\left\{\Ket{w_{l_v}}\right\}}_{l_v}$ of the Schmidt decomposition with respect to edge $e=\{p(v),v\} $ for each nonroot vertex $v \neq v_1$.
We write the computational basis of $\mathcal{H}_t^v$ as ${\{\Ket{l}\}}_{l}$ for $v\in V$.  Based on the recursive decomposition of $\mathcal{H}_t$ given by Eq.~\eqref{eq:h_each}, we represent $\Ket{\psi}$ as shown in the following proposition.
\begin{proposition}[Recursive description of multipartite states according to trees]
\label{lem:decomposition_tree}
For any rooted tree $(T,v_1)$ where $T=(V, E)$ and any $N$-partite pure state $\Ket{\psi} \in \mathcal{H}_t$, $\Ket{\psi}$ is decomposed by recursively applying the following procedure.
\begin{enumerate}
  \item For the root $v_1$, $\Ket{\psi}$ is decomposed in terms of the computational basis ${\{ \Ket{l} \}}_l$ of $\mathcal{H}_t^{v_1}$ and the Schmidt basis ${\{ \Ket{w_{l_c}} \}}_{l_c}$ of $\mathcal{H}_t^{D'_c}$ for all $c \in C_{v_1}$ as
\begin{equation}
    \begin{split}
        \Ket{\psi} = \sum_{l} \alpha^{v_{1}}_l \Ket{l}  \otimes \sum_{{\boldsymbol{l}_{C_{v_1}}}} \beta_{l,\boldsymbol{l}_{C_{v_1}}}^{v_1}  \bigotimes_{c\in C_{v_1}}\Ket{w_{l_c}},
    \end{split}
    \label{eq:psi_decomposition}
\end{equation}
where $l \in \{ 0, 1, \ldots, d_{v_1}-1 \}$ and 
$\boldsymbol{l}_{C_{v_1}} \coloneqq {(l_c)}_{c\in C_{v_1}}$ denotes a tuple of $l_c \in \{ 0, 1, \ldots , R_{\{v_1, c \} }-1 \}$.

\item For any nonleaf $v$, each Schmidt basis $\Ket{w_{l_v}} \in \mathcal{H}_t^{D'_v}$ in Eq.~\eqref{eq:psi_decomposition} is decomposed in terms of the computational basis ${\{ \Ket{l} \}}_l$ of $\mathcal{H}_t^{v}$ and the Schmidt basis states ${\{ \Ket{w_{l_c}} \}}_{l_c}$ of $\mathcal{H}_t^{D'_c}$ for all $c \in C_{v}$ as
\begin{equation}
    \begin{split}
        \Ket{w_{l_v}} =& \sum_{l}\alpha^{v}_{l, l_v} \Ket{l}  \otimes 
        \sum_{\boldsymbol{l}_{C_v}}
        \beta_{l,\boldsymbol{l}_{C_v},l_v}^{v}  \bigotimes_{c\in C_v}\Ket{w_{l_c}},
    \end{split}
    \label{eq:w_decomposition}
\end{equation}
where $l \in \{ 0, 1, \ldots, d_{v}-1 \}$ and $\boldsymbol{l}_{C_{v}} \coloneqq {(l_c)}_{c\in C_{v}}$ denotes a tuple of $l_c \in \{ 0, 1, \ldots , R_{\{v, c \}}-1 \}$.
\end{enumerate}
\end{proposition}

We first describe an overview of the state construction algorithm.  The algorithm consists of the following three elements.
\begin{enumerate}
\item Sequential measurements of the initial resource state $\ket{\Phi_{\textup{res}}(T)}$ performed by the parties except at leaves $v_{\textup{leaf}}$ of the rooted tree $(T,v_1)$ in the descending order starting from the root $v_1$.
\item Stepwise corrections depending on the measurement outcome of each parent $p(v)$ applied before performing their own measurements by the parties $v \neq v_1$ to compensate randomness induced by parent's measurement.
\item Final isometry transformations performed by the parties at leaves $v_{\textup{leaf}}$ to adjust bases from the computational basis to the Schmidt basis of the target state $\ket{\psi}$.
\end{enumerate}

\begin{figure}
\centering
\includegraphics[width=8cm]{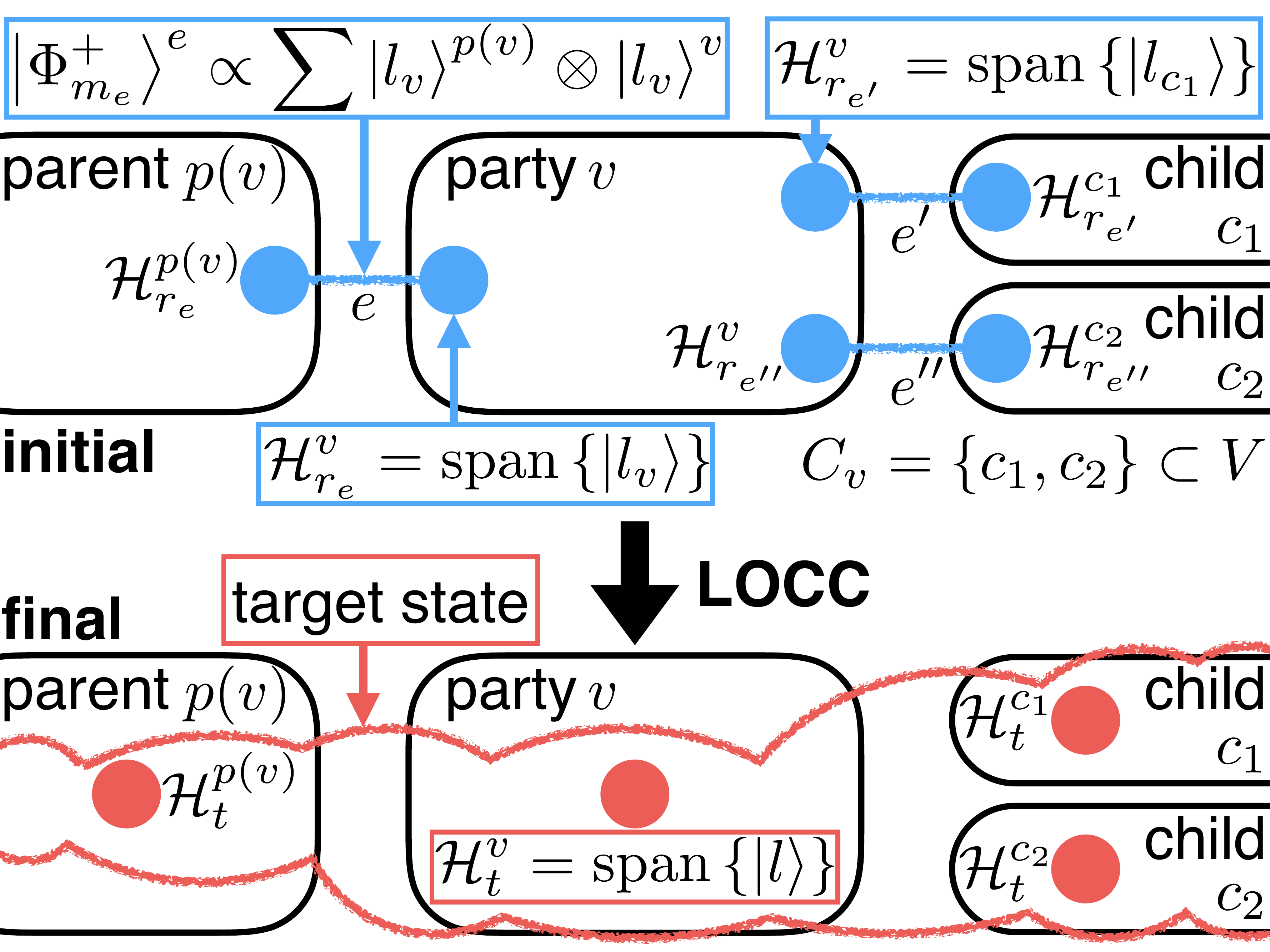}
\caption{(Color online) Schematic representation for the notations of the Hilbert spaces of the initial resource states and the target states.
    For each nonroot $v$ ($v\neq v_1$), an initial resource state shared between $v$ and its parent $p(v)$ is denoted by $\Ket{\Phi^+_{m_e}}^e=\left(1/\sqrt{m_e}\right)\sum_{l_v}\Ket{l_v}^{p(v)}\otimes\Ket{l_v}^v$, where $e=\{p(v),v\}$, $\Ket{l_v}^{p(v)}\in\mathcal{H}^{p(v)}_{r_e}$, and $\Ket{l_v}^v\in\mathcal{H}^v_{r_e}$.
    For each $v$, the computational basis of $\mathcal{H}^v_t$ for the $v$'s part of the target state is denoted by ${\{\Ket{l}\}}_l$.
}
\label{fig:system}
\end{figure}

In the following, we provide a general description of the above three elements in our algorithm.
Note that, as we will see through examples in the next subsection, the measurements and corrections in our algorithm are equivalent to performing sequential quantum teleportation starting from the root by using a part of the initial resource state, combined with each party's unitary transformation on another part of the initial resource system, an auxiliary system and a target system.
To distinguish computational basis states of different Hilbert spaces of $\ket{\Phi_{\textup{res}}(T)}$, we denote the computational basis of $\mathcal{H}^v_{r_e}$ where $e=\{p(v),v\}$ by ${\{\Ket{l_v}\}}_{l_v}$.     The initial resource state on $\mathcal{H}^e_r=\mathcal{H}^{p(v)}_{r_e}\otimes\mathcal{H}^v_{r_e}$ shared over an edge $e=\{p(v),v\}$ is written as
\[
    \Ket{\Phi^+_{R_e}}^e=\frac{1}{\sqrt{R_e}}\sum_{l_v=0}^{R_e-1}\Ket{l_v}^{p(v)}\otimes\Ket{l_v}^v,
\]
where $\Ket{l_v}^{p(v)}\in \mathcal{H}^{p(v)}_{r_e}$ and $\Ket{l_v}^v\in\mathcal{H}^v_{r_e}$.
Notations of the Hilbert spaces of the initial resource states and the target state are summarized in Fig.~\ref{fig:system}.

The measurement of each party at nonleaf $v \neq v_{\textup{leaf}}$ are chosen to transform $\ket{\Phi_{\textup{res}}(T)}$ into the reduced target state represented by the recursive description of $\Ket{\psi}$ given by Proposition~\ref{lem:decomposition_tree}.  The measurement of $v$ is represented by a set of measurement operators
\[
     {\left\{M^v_{\boldsymbol{x}_{C_v},\boldsymbol{z}_{C_v}}\right\}}_{\boldsymbol{x}_{C_v},\boldsymbol{z}_{C_v}}
\]
on $\mathcal{H}_r^{v}$  where $\boldsymbol{x}_{C_v} \coloneqq{(x_c)}_{c\in C_v}$ and $\boldsymbol{z}_{C_v} \coloneqq{(z_c)}_{c\in C_v}$ are tuples of the labels of the measurement outcomes where $x_c, z_c \in \{ 0,1,\ldots,R_{\{v,c\}}-1\}$ for each $c\in C_v$. 

For the root $v_1$, $M^{v_1}_{\boldsymbol{x}_{C_{v_1}},\boldsymbol{z}_{C_{v_1}}}$ is given by
\begin{equation}
     \label{eq:POVM_root}
     \begin{split}
     M^{v_1}_{\boldsymbol{x}_{C_{v_1}},\boldsymbol{z}_{C_{v_1}}}  =&\sum_{l} \alpha^{v_1}_l\Ket{l}\left(\sum_{\boldsymbol{l}_{C_{v_1}}}\beta_{l,\boldsymbol{l}_{C_{v_1}}}^{v_1}\right.\\
         &\left.\bigotimes_{c\in C_{v_1}}
     \left(\Bra{l_c}Z(z_c)X(x_c)/\sqrt{R_{\{v_1,c\}}}\right)
                 \right),
     \end{split}
\end{equation}
where $\alpha^{v_1}_l$ and $\beta_{l,\boldsymbol{l}_{C_{v_1}}}^{v_1}$ are the coefficients of the target state $\Ket{\psi}$ given by Eq.~\eqref{eq:psi_decomposition}.  $X(x_c)$ and $Z(z_c)$ in Eq.~\eqref{eq:POVM_root} are the generalized Pauli operators (also called Heisenberg-Weyl operators)~\cite{Wilde} transforming a computational basis state $\Ket{l_c}$ as 
\begin{equation}
    X(x_c) \Ket{l_c}=\Ket{(l_c+x_c)\mod R_e}
    \label{eq:HWopx}
\end{equation}
and
\begin{equation}
    Z(z_c) \Ket{l_c}=\exp \left [2\pi i  z_c l_c / R_e \right] \Ket{l_c},
 \label{eq:HWopz}
\end{equation}
where  $e=\{p(c),c\}$, respectively.  Note that the dimensions of these operators depend on the rank $R_e$.

For party $v$ where $v \neq v_1$ and $v \neq v_{\textup{leaf}}$, each $M^v_{\boldsymbol{x}_{C_v},\boldsymbol{z}_{C_v}}$ is given by
\begin{equation}
    \label{eq:POVM_each}
    \begin{split}
        M^v_{\boldsymbol{x}_{C_v},\boldsymbol{z}_{C_v}}  =&\sum_{l_v}\sum_{l} \alpha^{v}_{l,l_v}\Ket{l}\left(\Bra{l_v}\otimes\sum_{\boldsymbol{l}_{C_v}}\beta_{l,\boldsymbol{l}_{C_v},l_v}^{v}\right.\\
        &\left.\bigotimes_{c\in C_v}
        \left(\Bra{l_c}Z(z_c)X(x_c)/\sqrt{R_{\{v,c\}}}\right)
        \right),
\end{split}
\end{equation}
where $l_v \in \{0,1,\ldots,R_{\{p(v),v\}}-1 \}$, $\alpha^{v}_l$ and $\beta_{l,\boldsymbol{l}_{C_{v}}}^{v}$ are given by Eq.~\eqref{eq:w_decomposition}, other notations are same as the ones in Eq.~\eqref{eq:POVM_root}.  The completeness of each measurement is proven in Appendix~\ref{app:a}.

After performing the measurement, each party $v$ sends to each child $c\in C_v$ the measurement outcome $(x_c, z_c)$ by classical communication.   The child $c$ applies a stepwise generalized Pauli correction $Z(-z_c)X(x_c)$ on $\mathcal{H}^c_{r_{\{v,c\}}}$ depending on $(x_c, z_c)$ to remove the extra random transformation ${\left(Z(z_c)X(x_c)\right)}^T$, where $T$ represents transposition, induced by parent's measurement.   And if not at a leaf, the child $c$ recursively performs the measurement and then send the measurement outcome by classical communication to its children.  

At each leaf $v_{\textup{leaf}}$, after performing its generalized Pauli correction, the party performs the isometry transformation $U_{v_{\textup{leaf}}}$ transforming the computational basis ${\{ \ket{l_{v_{\textup{leaf}}}} \}}_{l_{v_{\textup{leaf}}}}$ of $\mathcal{H}^{v_{\textup{leaf}}}_r$ into the Schmidt basis ${\{ \ket{w_{l_{v_{\textup{leaf}}}}} \}}_{l_{v_{\textup{leaf}}}}$ of $\mathcal{H}^{v_{\textup{leaf}}}_t$ to adjust the bases at the leaves, which completes distributed construction of $\Ket{\psi}$ represented by Eqs.~\eqref{eq:psi_decomposition} and~\eqref{eq:w_decomposition}.

The step by step description of the distributed algorithm for exact state construction is presented in the proof of Theorem~\ref{thm:exact} as Algorithm~\ref{alg:1} for the party at $v_1$, Algorithm~\ref{alg:2} for the parties $v \neq v_1, v_{\textup{leaf}}$, and Algorithm~\ref{alg:3} for the parties $v_{\textup{leaf}}$ in Appendix~\ref{app:a}.

\begin{figure}
\centering
\includegraphics[width=8cm]{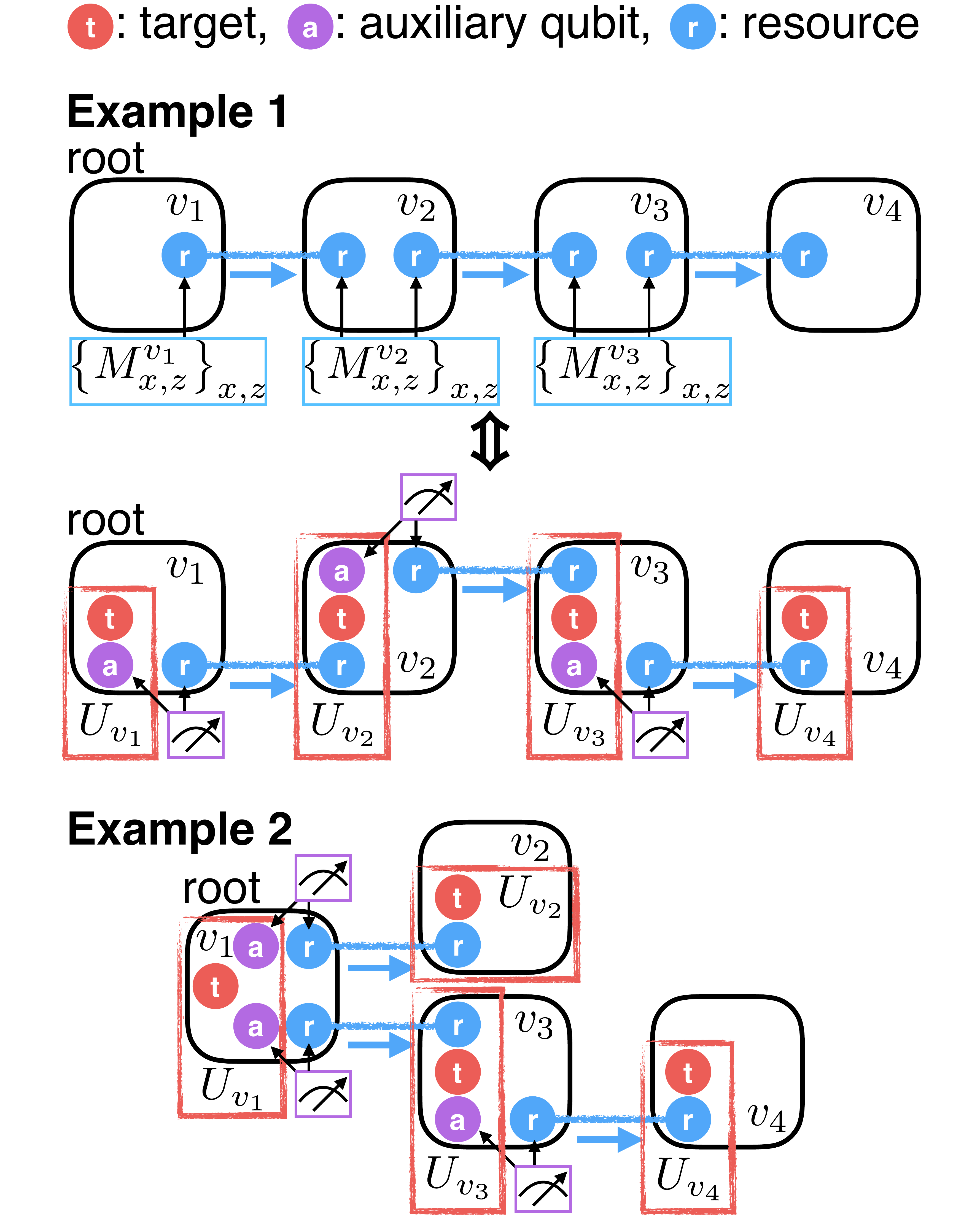}
\caption{(Color online) Schematic description of examples for exact state construction of the four-qubit $W$ state $\Ket{W_4}$ over four parties connected by a line-topology graph.  Theorem~\ref{thm:exact} implies that one Bell state $\ket{\Phi^+_2}$ (a maximally entangled two-qubit state) is required at each edge for the initial resource state in these examples.  The measurement for a party $v$ on an initial resource qubit (represented as ${\left\{M^v_{x,z}\right\}}_{x,z}$ in a blue rectangle) can be implemented by a unitary transformation $U_v$ (represented as a red rectangle) on the incoming initial resource qubit, an auxiliary qubit and a target qubit followed by the Bell measurement (represented as a purple rectangle) on the auxiliary qubit and the outgoing initial resource qubit for quantum teleportation.
Example~\ref{ex:1} is a case in which each party has at most one child, but the child may have descendants.
Example~\ref{ex:2} describes a case in which a party $v_1$ has multiple children $v_2$ and $v_3$.}
\label{fig:example}
\end{figure}

\subsection{Examples of exact state construction for four parties}
To demonstrate how our algorithm works, we provide examples of exact state construction.
We consider a line-topology graph connecting four parties,
and demonstrate exact construction of the four-qubit $W$ states 
\[
    \Ket{W_4}\coloneqq\frac{1}{2}\left(\Ket{1000}+\Ket{0100}+\Ket{0010}+\Ket{0001}\right)
\]
over the graph for two different choices of the root.  For construction of the $W$ states,  Theorem~\ref{thm:exact} implies that one Bell state $\ket{\Phi^+_2}^e$ (a maximally entangled two-qubit state) is required at each edge $e$, i.e.,\
\[
  E_{\textup{GC},i,e}^T(W_4) = 1
\]
for any tree $T$.   Thus, the required initial resource state is independent of the choice of the root of the graph, while the algorithm depends on the choice.  The examples corresponding to the two choices of the root represented in Fig.~\ref{fig:example} illuminate two essential ingredients in our algorithm for the multipartite state construction: one is when a child of a vertex has descendants, and the other is when a vertex has multiple children.

In Eqs.~\eqref{eq:POVM_root} and~\eqref{eq:POVM_each}, the measurement for a party $v$ in our algorithm is represented as a single measurement operator $M^v_{\boldsymbol{x}_{C_v},\boldsymbol{z}_{C_v}}$ on the \textit{whole} initial resource qubits at $v$ given by $\mathcal{H}_r^v = \bigotimes_{e} \mathcal{H}_{r_e}^v$.
By introducing an auxiliary system $\mathcal{H}_a^v$ and an target system $\mathcal{H}_t^v$ in fixed states at each party $v$,
$M^v_{\boldsymbol{x}_{C_v},\boldsymbol{z}_{C_v}}$ can be implemented by a unitary transformation $U_v$ on 
the auxiliary qubit, the target qubit and a part of the initial resource qubit at $v$ given by $\mathcal{H}_{r_{\{p(v),v\}}}^v$ referred to as an \textit{incoming} initial resource qubit, followed by a Bell measurement for quantum teleportation on the auxiliary qubit and another part of the initial resource qubits at $v$ given by $\bigotimes_{c\in C_v}\mathcal{H}_{r_{\{ v, c\} }}^v$ referred to as \textit{outgoing} initial resource qubits, as illustrated in Example 1 of Fig.~\ref{fig:example}.

We first describe an example in which each vertex has at most one child but the child may have descendants.
\begin{example}[\label{ex:1}Exact state construction with a child having descendants]
Consider exact construction of the four-qubit $W$ state $\Ket{W_4}$ over four parties represented by $V = \{v_1, v_2, v_3, v_4\}$.
The parties are connected by a line-topology graph $T=(V, E)$ illustrated at the top of Fig.~\ref{fig:example}, where $E=\left\{\left\{v_i, v_{i+1}\right\}:i=1,2,3\right\}$.
We choose $v_1$ as the root so that each vertex has at most one child.

First, we show the operation for $v_1$.
Using Proposition~\ref{lem:decomposition_tree}, we can decompose the target state $\Ket{W_4}$ as
\begin{align*}
  &\Ket{W_4}_t^{1,2,3,4}\\
  &\propto\Ket{1}_t^1\Ket{000}_t^{2,3,4} +\Ket{0}_t^{1}\left[\Ket{100}_t^{2,3,4}+\Ket{010}_t^{2,3,4}+\Ket{001}_t^{2,3,4}\right]\\
  &\in\mathcal{H}_t^{v_1}\otimes\mathcal{H}_t^{D'_{v_2}},
\end{align*}
where $D'_{v_2}=\{v_2,v_3,v_4\}$, the superscripts for each ket represent parties to which the ket belongs, and the subscripts are to identify targets, resources and auxiliary qubits.
Schmidt basis states $\Ket{000}_t^{2,3,4}$ and  $\Ket{100}_t^{2,3,4}+\Ket{010}_t^{2,3,4}+\Ket{001}_t^{2,3,4}$ can be encoded in an auxiliary qubit $\mathcal{H}_a^{v_1}$ on $v_1$ as $\Ket{0}_a^1$ and $\Ket{1}_a^1$, respectively.
Therefore, $v_1$ prepares $\Ket{0}_t^1\Ket{0}_a^1+\Ket{1}_t^1\Ket{1}_a^1$ by a unitary transformation $U_{v_1}$ from a fixed state $\ket{0}_t^1  \ket{0}_a^1$ and sends the auxiliary qubit state in $\mathcal{H}_a^{v_1}$ to $v_2$ by quantum teleportation using the Bell state $\ket{\Phi_2^+}_{r_{\{v_1,v_2\}}}^{1,2}$ consisting of the initial resource state. The state preparation followed by quantum teleportation is equivalent to the measurement of $v_1$ given by Eq.~\eqref{eq:POVM_root} on the outgoing initial resource qubit $\mathcal{H}_{r_{\{v_1,v_2\}}}^{v_1}$  followed by the correction on the incoming initial resource qubit $\mathcal{H}_{r_{\{v_1,v_2\}}}^{v_2}$  at $v_2$.  There is no incoming initial resource qubit at the root.

Next, we show the operation for $v_2$.
Using Proposition~\ref{lem:decomposition_tree} recursively, we obtain
\begin{align*}
&\Ket{W_4}_t^{1,2,3,4}\\
&\propto\Ket{1}_t^{1}\Ket{000}_t^{2,3,4} +\Ket{0}_t^{1}\left[\Ket{100}_t^{2,3,4}+\Ket{010}_t^{2,3,4}+\Ket{001}_t^{2,3,4}\right]\\
&=\Ket{1}_t^{1}\Ket{0}_t^{2}\Ket{00}_t^{3,4}\\
&\qquad+\Ket{0}_t^{1}\left[\Ket{1}_t^{2}\Ket{00}_t^{3,4}+\Ket{0}_t^{2}\left(\Ket{10}_t^{3,4}+\Ket{01}_t^{3,4}\right)\right]\\
&\in\mathcal{H}_t^{v_1}\otimes\mathcal{H}_t^{v_2}\otimes\mathcal{H}_t^{D'_{v_3}},
\end{align*}
where $D'_{v_3}=\{v_3, v_4\}$.
Similarly to the operation for $v_1$, $\Ket{00}_t^{3,4}$ and $\Ket{10}_t^{3,4}+\Ket{01}_t^{3,4}$ can be encoded in an auxiliary qubit $\mathcal{H}_a^{v_2}$ on $v_2$ as $\Ket{0}_a^2$ and $\Ket{1}_a^2$, respectively.
Then $v_2$ transforms the incoming resource qubit states $\Ket{0}^2_{r_{\{v_1,v_2\}}}$ and $\Ket{1}^2_{r_{\{v_1,v_2\}}}$ {encoding} $\Ket{0}_t^2\Ket{00}_t^{3,4}$ and $\Ket{1}_t^2\Ket{00}_t^{3,4}+\Ket{0}_t^2(\Ket{10}_t^{3,4}+\Ket{01}_t^{3,4})$ into $\Ket{0}_t^2\Ket{0}_a^2$ and $\Ket{1}_t^2\Ket{0}_a^2+\Ket{0}_t^2\Ket{1}_a^2$, respectively by a unitary transformation $U_{v_2}$ by introducing the auxiliary qubit and the target qubit, and sends the state on $\mathcal{H}_a^{v_2}$ to $v_3$ by quantum teleportation using the Bell state $\ket{\Phi_2^+}_{r_{\{v_2,v_3\}}}^{2,3}$  consisting of the initial resource state.
The state preparation followed by quantum teleportation is equivalent to the measurement of $v_2$ given by Eq.~\eqref{eq:POVM_each} on the initial resource qubit $\mathcal{H}_r^{v_2}$ followed by the correction on the incoming initial resource qubit $\mathcal{H}_{r_{\{v_2,v_3\}}}^{v_3}$ at $v_3$.

Finally, $v_3$ transforms the incoming resource qubit states $\Ket{0}_{r_{\{v_2,v_3\}}}^3$ and $\Ket{1}_{r_{\{v_2,v_3\}}}^3$ into $\Ket{0}_t^3\Ket{0}_a^3$ and $\Ket{1}_t^3\Ket{0}_a^3+\Ket{0}_t^3\Ket{1}_a^3$, respectively, and sends the state on $\mathcal{H}_a^{v_3}$ to $v_4$ by quantum teleportation using the Bell state $\ket{\Phi_2^+}_{r_{\{v_3,v_4\}}}^{3,4}$  consisting of the initial resource state.  Once the correction on the incoming resource qubit $\mathcal{H}_{r_{\{v_3,v_4\}}}^{v_4}$ at $v_4$ is completed, the incoming resource qubit turns out to be the target qubit state of $v_4$.

\end{example}

The next example is a case in which a vertex may have multiple children.
\begin{example}[\label{ex:2}Exact state construction with multiple children]
Similarly to Example~\ref{ex:1}, consider exact construction of $\Ket{W_4}$ under a line-topology graph $T'=(V, E')$ illustrated at the bottom of Fig.~\ref{fig:example}, where $V = \{v_1, v_2, v_3, v_4\}$ and $E'=\left\{\left\{v_1, v_2\right\}, \left\{v_1, v_3\right\}, \left\{v_3, v_4\right\}\right\}$.
Now, the root $v_1$ has two children $v_2$ and $v_3$.

We show the operation for $v_1$.
Using Proposition~\ref{lem:decomposition_tree}, we obtain
\begin{align*}
  &\Ket{W_4}_t^{1,2,3,4}\\
  &\propto\Ket{1}_t^1\Ket{0}_t^2\Ket{00}_t^{3,4}+\Ket{0}_t^1\Ket{1}_t^2\Ket{00}_t^{3,4}\\
  &\qquad+\Ket{0}_t^1\Ket{0}_t^2(\Ket{10}_t^{3,4}+\Ket{01}_t^{3,4})\\
  &\in\mathcal{H}_t^{v_1}\otimes\mathcal{H}_t^{D'_{v_2}}\otimes\mathcal{H}_t^{D'_{v_3}},
\end{align*}
where $D'_{v_2}=\{v_2\}$ and $D'_{v_3}=\{v_3,v_4\}$, and the notations of the superscripts and subscripts for each state are the same as the ones in Example~\ref{ex:1}.
In the same way as Example~\ref{ex:1}, we can encode $\Ket{00}_t^{3,4}$ and $\Ket{10}_t^{3,4}+\Ket{01}_t^{3,4}$  as $\Ket{0}_a^1$ and $\Ket{1}_a^1$, respectively.
We also encode $\Ket{0}_t^2$ and $\Ket{1}_t^2$ using another auxiliary qubit $\mathcal{H}_{a'}^{v_1}$ at $v_1$ as $\Ket{0}_{a'}^1$ and $\Ket{1}_{a'}^1$ respectively.
Therefore $v_1$ prepares $\Ket{1}_t^1\Ket{0}_{a'}^1\Ket{0}_a^1 + \Ket{0}_t^1\Ket{1}_{a'}^1\Ket{0}_a^1+\Ket{0}_t^1\Ket{0}_{a'}^1\Ket{1}_a^1$, and sends the states on $\mathcal{H}_{a'}^{v_1}$ and $\mathcal{H}_a^{v_1}$ to $v_2,v_3$ by quantum teleportation using the Bell states $\ket{\Phi_2^+}_{r_{\{v_1,v_2\}}}^{1,2}$ and $\ket{\Phi_2^+}_{r_{\{v_1,v_3\}}}^{1,3}$ consisting of the initial resource states, respectively.   After completing quantum teleportation, the incoming initial resource qubit of $v_2$ turns out to be in the target state.   The operations for $v_3$ and  $v_4$ are the same as the ones for $v_3$ and $v_4$ in Example~\ref{ex:1}.
\end{example}

\subsection{Connection with projected entangled pair states}
The presented algorithm for distributed state construction is \textit{deterministic} in contrast with a \textit{stochastic} LOCC (SLOCC) algorithm obtained from a state description by projected entangled pair states (PEPS)~\cite{PEPS}.
PEPS, including matrix product states (MPS) and tree tensor networks (TTN) as special cases, provides efficient classical description of a class of quantum multipartite states,
in which parameters of the states are expressed in terms of multiple tensors with their indices contracted according to a given network pattern.  A PEPS is a state in the form of
\begin{equation}
    \label{eq:PEPS}
    \left(\bigotimes_{v\in V}A^v\right)\bigotimes_{e\in E}\Ket{\Phi_{m_e}^+}^e,
\end{equation}
where graph $G=(V,E)$ is a line-topology graph for MPS, a tree for TTN and any graph for PEPS, and $A^v$ is a linear operator acting on the part of the bipartite maximally entangled states belonging to $v$, which corresponds to each tensor in the tensor network.

State description by PEPS does not necessarily provide an efficient method for distributed construction of the state with high success probability, while a number of algorithms in various settings are presented for constructing states described in PEPS using quantum computers~\cite{PRL110503,PRL110502,PRA032321,PRL080503}.  For state construction under a graph $G$, the description of the target state in the form of Eq.~\eqref{eq:PEPS} provides an SLOCC algorithm, in which each party $v\in V$ performs a binary outcome measurement of the initial resource state $\Ket{\Phi_{\textup{res}}(G)}=\bigotimes_{e\in E}\Ket{\Phi_{m_e}^+}^e$ given by $\left\{A^v, \sqrt{\openone-A^{v\dag} A^v}\right\}$, where $\openone$ denotes the identity operator, to transform $\Ket{\Phi_{\textup{res}}(G)}$ into the target state Eq.~\eqref{eq:PEPS} with nonzero probability.
Note that, in the SLOCC algorithm, classical communication is not necessary and parties can independently perform their measurement to achieve nonzero success probability, while classical communication may increase the success probability.

In contrast, in our LOCC setting, we achieve deterministic and exact construction of the target state by choosing the measurement operators so that all the measurements are correctable.  This means that a target state described by PEPS with tree networks (namely, TTN) is deterministically constructible by introducing an appropriate measurement order. The canonical form of the TTN is uniquely determined and can be efficiently calculated~\cite{PRA022320}.  The measurement operators used in our algorithm for exact state construction can be expressed in terms of the tensors in the canonical form of the TTN by modifying the tensors according to Proposition~\ref{lem:decomposition_tree}.   In Appendix~\ref{app:demo}, we show an explicit description of the measurement operators for a state given by the canonical form of a special case of TTN given by a line-topology graph, namely MPS~\cite{APractical}.

We remark that extension of our algorithm to arbitrary PEPS is not straightforward.
One reason for this difficulty is that there exists a multipartite state which satisfies the Schmidt rank condition with respect to any connected bipartite cut but cannot be expressed in terms of PEPS with its bond dimension satisfying the same condition, as we see through a counter example (Remark~\ref{ex:3}) presented in the next subsection.

\subsection{Implications of exact state construction}

We remark implications of exact state construction.

For exact construction of the $N$-qubit $W$ states defined by
\[
    \Ket{W}\coloneqq\frac{1}{\sqrt{N}}\left(\Ket{100\cdots 0}+\Ket{010\cdots 0}+\cdots+\Ket{000\cdots 1}\right)
\]
and the $N$-qubit GHZ states
\[
  \Ket{\textup{GHZ}}\coloneqq\frac{1}{\sqrt{2}}\left(\Ket{00\cdots 0}+\Ket{11\cdots 1}\right),
\]
using initial resource states represented by trees achieves the smallest exact total graph-associated entanglement cost among any graphs.
\begin{remark}[Exact construction of $W$ and GHZ states]
By applying Theorem~\ref{thm:exact} to the $N$-qubit $W$ states and GHZ states over $N$ parties, we can derive the smallest exact total graph-associated entanglement cost among \textit{any} graph representing the restriction on quantum communication among the $N$ parties.
The $N$-qubit $W$ and GHZ states over $N$ parties have the Schmidt rank of $2$ with respect to any bipartition.
Therefore, Theorem~\ref{thm:exact} implies that, under any tree $T$, the $N$-qubit $W$ and GHZ states can be constructed with consumption of one Bell state $\ket{\Phi_2^+}$ at each edge,
i.e.,\ 
  \[
  E_{\textup{GC},i,e}^T = 1.
\]

Note that any connected graph contains a spanning tree, which can be used for state construction.
Since any tree has $N-1$ edges, the $N$-qubit $W$ states and GHZ states can be constructed with consumption of $N-1$ Bell states in total.
This provides the smallest exact total graph-associated entanglement cost among any graphs for exact construction of the $N$-qubit $W$ and GHZ states, since any party has to consume at least one Bell state initially shared with another.  The similar optimality can be shown for any genuinely entangled states whose Schmidt rank with respect to any bipartition is $2$.
\end{remark}

In comparison with the state distribution algorithm in which the root prepares the whole target state and distributes the corresponding part of the state to each party, the following remark illustrates quadratic saving of total resource consumption in exact construction of the same target state.
\begin{remark}[Advantage of exact state construction in comparison with simple distribution]
    Consider preparation of the $N$-qubit $W$ state over $N$ parties represented by $\{v_1,v_2,\ldots,v_N\}$, where the initial resource state is in a line topology connecting $\{v_i,v_{i+1}\}$ for all $i\in\{1,2,\ldots,N-1\}$.
    Even if we choose $v_1$ as the root, which is placed at an end of the line topology, exact construction of the $W$ state can be performed with consumption of one Bell state at each edge, which amounts to $N-1$ Bell states in total.
    In contrast, if $v_1$ prepares the $N$-qubit $W$ state on its own, then to distribute the $W$ state from one party to another by quantum teleportation consuming the initial resources, $N-i$ Bell states are consumed at the edge $\{v_i,v_{i+1}\}$ for all $i\in\{1,2,\ldots,N-1\}$, which amounts to $\sum_{i=1}^{N-1}\left(N-i\right)=(1/2)N(N-1)$ Bell states in total.
    The difference in this example is quadratic with respect to the number of the parties $N$.
\end{remark}

We also remark that, for a nontree graph, the Schmidt rank of a target state with respect to a bipartite cut does not tell whether the target state can be exactly constructible from the initial resource state, as we show a counter example.
\begin{remark}[\label{ex:3}Exact state construction not under trees]
    Given four parties $v_1$, $v_2$, $v_3$ and $v_4$, one counter example is the case when the initial resource state consists of four Bell states shared between $v_1$--$v_2$, $v_2$--$v_3$, $v_3$--$v_4$ and $v_4$--$v_1$, whose underlying graph contains a cycle, and the target state two Bell states shared between $v_1$--$v_3$ and  $v_2$--$v_4$.  For a bipartite cut $\{v_1,v_2\}$ and $\{v_3,v_4\}$, the Schmidt rank of the initial resource state is $4$.  The Schmidt rank of the target state with respect to the bipartite cut $\{v_1,v_2\}$ and $\{v_3,v_4\}$ is also $4$.  However, this exact state construction has been proven to be impossible even probabilistically in Ref.~\cite{RefWorks:167}.
\end{remark}

\section{\label{sec:4}Approximate state construction under trees}

\subsection{Graph-associated entanglement cost for approximate state construction}

We present a second-order asymptotic analysis of $(n,\epsilon)$-approximate construction tasks under trees.   In bipartite cases, the entanglement cost in second-order asymptotic analysis can be expressed in terms of the quantum information spectrum entropy~\cite{RefWorks:160}.
The quantum information spectrum entropy is defined for any density operator $\rho$ as
\[
    \overline{H}_{\textup{S}}^{\epsilon}\left(\rho\right)\coloneqq \inf\left\{\gamma\in\mathbb{R}\colon\tr{\left(\rho - 2^{- \gamma}\openone\right)} _+ \geqq 1-\epsilon\right\},
\]
where ${(\cdot)}_+$ represents the operator projected onto the non-negative spectra.
Given a pure state $\Ket{\psi}$ shared between A and B, the $(n,\epsilon)$-approximate bipartite entanglement cost $E_{\textup{C}}^{n,\epsilon}\left(\psi\right)$ per construction of one state can be bounded by using $\overline{H}_{\textup{S}}^{\epsilon}$ and any fixed $\delta, \eta>0$, as
\begin{equation*}
    \begin{split}
        &\overline{H}_{\textup{S}}^{\epsilon^2/4+\eta}\left(\rho_A^{\otimes n}\right)-\delta+\log_2 \eta\\
        &\leqq n E_{\textup{C}}^{n,\epsilon}\left(\psi\right) \leqq \overline{H}_{\textup{S}}^{\epsilon^2/4}\left(\rho_A^{\otimes n}\right),
\end{split}
\end{equation*}
where $\rho_A=\tr_B \psi$, and the error threshold $\epsilon$ for approximation is evaluated by the trace distance.  We note that, in Ref.~\cite{RefWorks:160}, the fidelity is used to evaluate the error threshold for approximation in place of the trace distance.

The quantum information spectrum entropy can be evaluated by numerical calculation of the smooth min-entropy by semidefinite programming, which gives a tight bound~\cite{HierarchyOfInformationQuantities}.   A good approximation of the quantity for large $n$ is obtained by the second-order expansion of the quantum information spectrum entropy~\cite{RefWorks:160}
\begin{equation*}
    \begin{split}
        \overline{H}_{\textup{S}}^{\epsilon^2/4}\left(\rho^{\otimes n}\right)=&a(\rho)n+b(\rho,\epsilon)\sqrt{n}+\mathcal{O}\left(\log n\right) \,(\textup{as } n\rightarrow\infty),\\
        a(\rho) =& S\left(\rho\right),\\
        b(\rho,\epsilon)=&-s\left(\rho\right)\Phi^{-1}\left(\epsilon^2/4\right),
\end{split}
\end{equation*}
where  $S(\rho)\coloneqq-\tr\rho\log_2\rho$ is the von Neumann entropy,
$s\left(\rho\right)= \sqrt{\tr\rho{\left(\log_2 \rho\right)} ^2 - {S\left(\rho\right)} ^2}$ is
the quantum information standard deviation, and
$\Phi\left(z\right)=\left(1/{\sqrt{2\pi}}\right)\int_{-\infty}^{z}e^{-t^2/2}dt$
the cumulative distribution function of a standard normal random variable.  The first-order coefficient $a(\rho)$ represents the asymptotic rate, while the second-order coefficient $b(\rho,\epsilon)$ is the difference from the asymptotic rate in the finite block length regime.

For approximate construction of a multipartite state $\ket{\psi}$ under tree $T=(V,E)$,  we evaluate the $(n,\epsilon)$-approximate edge graph-associated entanglement cost $E_{\textup{GC},i,e}^{n,\epsilon,T}\left(\psi\right)$ for each edge $e \in E$ by applying the second-order asymptotic analysis.    We derive upper and lower bounds of $E_{\textup{GC},i,e}^{n,\epsilon,T}\left(\psi\right)$ expressed in terms of the quantum information spectrum entropy of $n$ copies of the reduced density operator $\rho_e^{\otimes n}$ of $\ket{\psi}$ with respect to $e$.   We obtain the following theorem.

\begin{theorem}[Approximate edge graph-associated entanglement cost]
\label{thm:asymptotic}
Given any tree $T=(V,E)$ and any $N$-partite pure target state $\Ket{\psi}\in\mathcal{H}_t$,
fix $n>0$ and $\epsilon>0$.
The following bounds hold for any configuration $i$ of the optimal resource state for the $(n,\epsilon)$-approximate construction of $\Ket{\psi}$ under $T$.
\begin{enumerate}
    \item The upper bound: For any error thresholds at respective edges, denoted by $\epsilon'(e)>0$ for all $e\in E$, satisfying
\[
  \sqrt{\sum_{e\in E}{\epsilon'(e)}^2}\leqq\epsilon,
\]
it holds that
    \[
        \sum_{e\in E}E_{\textup{GC},i,e}^{n,\epsilon,T}\left(\psi\right) \leqq \sum_{e\in E}\frac{\overline{H}_{\textup{S}}^{{\epsilon'(e)}^2/4}\left(\rho_e^{\otimes n}\right)}{n}.
    \]
\item The lower bound: For each $e\in E$ and any $\delta,\eta>0$,
    \[
        E_{\textup{GC},i,e}^{n,\epsilon,T}{\left(\psi\right)}\geqq \frac{\overline{H}_{\textup{S}}^{\epsilon^2/4+\eta}\left(\rho_e^{\otimes n}\right)-\delta+\log_2 \eta}{n}.
    \]
\end{enumerate}
\end{theorem}

The asymptotic limit of the above theorem yields the following, which coincides with the one shown by Galv\~ao \textit{et al.}~\cite{RefWorks:148}.
\begin{corollary}[Asymptotic edge graph-associated entanglement cost]
\label{cor:asymptotic_cost}
Given any tree $T=(V,E)$ and any $N$-partite pure target state $\Ket{\psi}\in\mathcal{H}_t$, for each $e\in E$,
\[
E^{\infty, T}_{\textup{GC},i,e}(\psi)=S\left(\rho_e\right),
\]
where the configuration index $i$ is uniquely determined.
\end{corollary}

While the lower bound in Theorem~\ref{thm:asymptotic} is deduced from the bipartite cases as presented in Appendix~\ref{app:b}, the upper bound is nontrivial.  To prove the upper bound, we explicitly present an efficient distributed algorithm for approximate construction of a multipartite state under a tree.  In the asymptotic limit, the strategy sketched in Ref.~\cite{RefWorks:148} coincides with our algorithm.

Approximate construction of $n$ copies of the target state $\Ket{\psi}^{\otimes n}$ under tree $T$ can be achieved by exact construction of an approximate state $\Ket{\tilde\psi_n}$, which is $\epsilon$-close to $\Ket{\psi}^{\otimes n}$ in terms of the trace distance.   For each $e = \{p(v),v\}\in E$ of a rooted tree, we define the following projections on the Hilbert spaces of $n$ copies of the target state for approximation, which only act nontrivially on the Hilbert spaces at the vertices in $D'_v$,  the set of $v$'s/ descendants and $v$ itself, given by Eq.~\eqref{eq:v_each}.

\begin{definition}[Projections for approximation of multipartite states]
\label{def:projection}
Given a rooted tree $(T, v_1)$ where $T=(V,E)$, an $N$-partite pure state $\Ket{\psi}\in\mathcal{H}_t$, fixed $n>0$, and an error threshold $\epsilon'(e)>0$ for $e = \{p(v),v\}\in E$, we define a projection $\Pi_{\rho_e}^{n,\epsilon'(e)}$ onto the non-negative spectra of $\rho_e^{\otimes n}-2^{-\gamma_e}\openone$ of ${\left(\mathcal{H}^{D'_v}_t\right)}^{\otimes n}$, where $\rho_e$ is the reduced density operator of $\Ket{\psi}$ with respect to $e$ and $\gamma_e$ is given by the quantum information spectrum entropy $\gamma_e=\overline{H}_{\textup{S}}^{{\epsilon'(e)}^2/4}\left(\rho_e^{\otimes n}\right)$. We write a projection on ${(\mathcal{H}_t)}^{\otimes n}$  as $\Pi_{e,\psi}^{n,\epsilon'(e)}\coloneqq\Pi_{\rho_e}^{n,\epsilon'(e)}\otimes\openone$.  We omit $n,\epsilon'(e)$ in the superscripts of the projections to write $\Pi_{\rho_e}$ and $\Pi_{e,\psi}$ if obvious.
\end{definition}

A straightforward calculation shows that, for each $e\in E$, it holds that
\begin{equation}
    \label{eq:projection}
    \begin{split}
    \tr \Pi_{e,\psi}\psi^{\otimes n} &= \tr \Pi_{\rho_e}\rho_e^{\otimes n}\\
    &\geqq\tr\Pi_{\rho_e}{\left(\rho_e^{\otimes n}-2^{-\gamma_e}\openone\right)}_+\\
    &=1-\left({\epsilon'(e)}^2/4\right).
\end{split}
\end{equation}
We provide a method for obtaining a good approximation for a large number of identical copies of $\Ket{\psi}$ in the following proposition, which is an application of noncommutative union bound~\cite{pra052331}.  The proof of the proposition is presented in Appendix~\ref{app:b}.   
\begin{proposition}[Approximation of multipartite states]
\label{prp:approximation}
Given any rooted tree $(T,v_1)$ where $T=(V,E)$ and any $N$-partite pure state $\Ket{\psi}\in\mathcal{H}_t$, fix $n > 0$ and $\epsilon'\colon E\rightarrow (0,\infty)$.
Then, it holds that
\begin{equation}
\label{eq:approximation}
\begin{split}
    \left\|\psi^{\otimes n}-\tilde\psi_n\right\| _1 &\leqq\sqrt{\sum_{e\in E}{\epsilon'(e)}^2},
\end{split}
\end{equation}
where  $\tilde\psi_n=\Ket{\tilde\psi_n}\Bra{\tilde\psi_n}$ is defined by
\begin{equation}
\label{eq:approximated_state}
\Ket{\tilde\psi_n}\coloneqq\frac{\Pi_{e_{N-1},\psi}\Pi_{e_{N-2},\psi}\cdots\Pi_{e_1,\psi}\Ket{\psi}^{\otimes n}}{\left\|\Pi_{e_{N-1},\psi}\Pi_{e_{N-2},\psi}\cdots\Pi_{e_1,\psi}\Ket{\psi}^{\otimes n}\right\|}.
\end{equation}
\end{proposition}

As Proposition~\ref{prp:approximation} implies, an efficient distributed algorithm for approximate state construction is the one for exact construction of $\Ket{\tilde\psi_n}$ defined by Eq.~\eqref{eq:approximated_state}.  In our algorithm, parties do not perform the projections on their quantum system, but classically calculate the description of the approximate state according to Eq.~\eqref{eq:approximated_state}. 
The error threshold $\epsilon'$ for the edges is determined so that it minimizes
\[
    \sum_{e\in E}\overline{H}_{\textup{S}}^{{\epsilon'(e)}^2/4}\left(\rho_e^{\otimes n}\right),
\]
within the constraint of
\begin{equation}
\label{eq:constraints_epsilon}
\begin{split}
    &\sqrt{\sum_{e\in E}{\epsilon'(e)}^2}\leqq\epsilon.
\end{split}
\end{equation}
Note that, for any fixed $\rho_e$, $\overline{H}_{\textup{S}}^{{\epsilon'(e)}^2/4}\left(\rho_e^{\otimes n}\right)$ monotonically decreases as $\epsilon'$ increases.
For sufficiently large $n$, the second-order expansion of the quantum information spectrum entropy implies that $\epsilon'$ for efficient approximate state construction is approximately determined when it minimizes the sum of the second-order coefficients
\[
    \sum_{e\in E}b\left(\rho_e,\epsilon'(e)\right)
\]
within the constraint of Eq.~\eqref{eq:constraints_epsilon}.

Therefore, approximate construction of the state $\Ket{\psi}$ under any tree $T$ is achieved by first designating the root, and then by following the order of the rooted tree, the parties performs the Algorithms~\ref{alg:1},~\ref{alg:2},~\ref{alg:3} given in the proof of Theorem~\ref{thm:exact} for exact construction of $\Ket{\tilde\psi_n}$ classically calculated by Eq.~\eqref{eq:approximated_state}.

In calculating $\Ket{\tilde\psi_n}$, improvement of a solution of combinatorial optimization for $\epsilon'$ leads to an efficient distributed algorithm.  However, our result of the lower bound in Theorem~\ref{thm:asymptotic} for approximate state construction is not tight enough to prove its optimality, as we will show through examples in the next subsection.

\subsection{Examples of approximate state construction}
We illustrate examples of the approximate state construction.

We first present approximate construction of the $N$-qubit $W$ states over $N$ parties under a line-topology graph to compare with the case of exact state construction, which requires $E_{\textup{GC},i,e}^T(W) = 1$.   We show how approximate construction consumes less resources per copy by increasing the number of copies $n$ and the resource consumption rate varies in terms of the number of parties $N$.

\begin{figure}
\centering
\includegraphics[width=8cm]{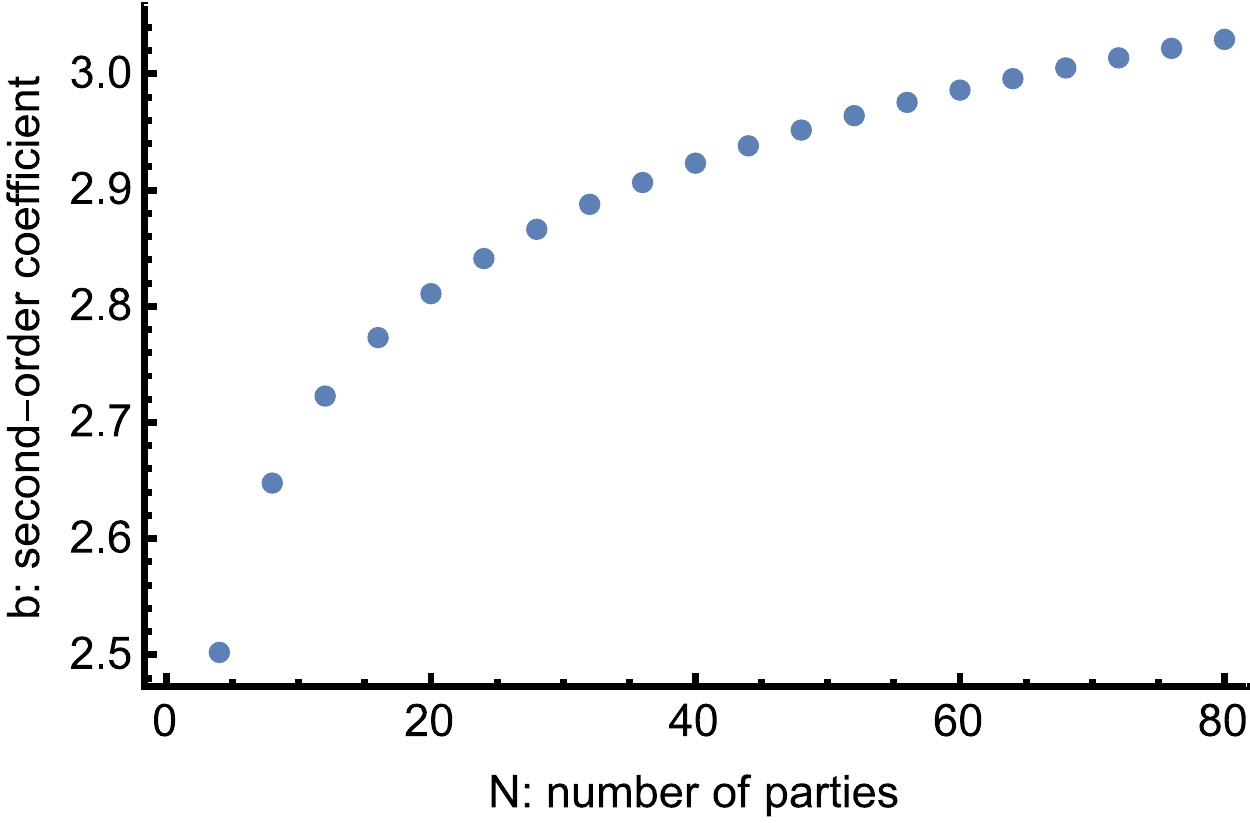}
\caption{The second-order coefficient $b$ of the resource consumption rate at edge $e_{N/4}$ of our approximate state construction algorithm for the $N$-qubit $W$ states over $N$ parties under line-topology graphs for $N\in\{4,8,\ldots,80\}$.   The error thresholds is set to be $\epsilon'(e)={(N-1)}^{-1/2}\epsilon$ for all $e\in E$ with $\epsilon=1/25$. The second-order coefficient increases as $N$ increases, whereas the first-order asymptotic rate is independent of $N$.}
\label{fig:graph3}
\end{figure}

\begin{example}[Approximate construction of the $W$ states]
\label{ex:4}
We consider approximate construction of the $N$-qubit $W$ states over $N$ parties under $G=(V,E)$ where $V=\left\{v_1,v_2,\ldots,v_N\right\}$ and $E=\left\{\{v_i,v_{i+1}\}\colon i\in\{1,2,\ldots,N-1\}\right\}$.
We set $\epsilon=1/25$, which guarantees that the optimal success probability to distinguish the approximately constructed state from the $N$-qubit $W$ state is smaller than $51$\%.
We assume that the error threshold for approximation with respect to each edge is a constant and $\epsilon'$ is given by $\epsilon'={(N-1)}^{-1/2}\epsilon$.  We calculate the resource consumption rate of the initial resource state at edge $e_{N/4}$ achieved by our algorithm for approximate state construction.  The resource consumption rate in our algorithm is expanded up to the second order as
\begin{equation*}
    \begin{split}
        \frac{\overline{H}_{\textup{S}}^{\epsilon'^2/4}\left(\rho_{e_{N/4}}^{\otimes n}\right)}{n}=&a\left(\rho_{e_{N/4}}\right)+b\left(\rho_{e_{N/4}}, \epsilon'\right)\frac{1}{\sqrt{n}}+\\
        &\mathcal{O}\left(\frac{\log n}{n}\right) \,(\textup{as } n\rightarrow\infty).
    \end{split}
\end{equation*}
The second-order coefficient $b\left(\rho_{e_{N/4}}, \epsilon'\right)$ for $N\in\{4,8,\ldots,80\}$ is shown in Fig.~\ref{fig:graph3}.
As illustrated, given a fixed error threshold $\epsilon$ in finite block length regime, more resources are needed for approximate construction of the $N$-qubit $W$ states for increasing $N$, while the first-order term
\[
  a\left(\rho_{e_{N/4}}\right)=0.811\cdots
\] representing the asymptotic rate is constant for any $N$.
\end{example}

We can further improve the resource consumption rate by individually choosing error threshold $\epsilon'(e)$ for each $e\in E$ optimized for given total error threshold $\epsilon$.

\begin{figure}
\centering
\includegraphics[width=8cm]{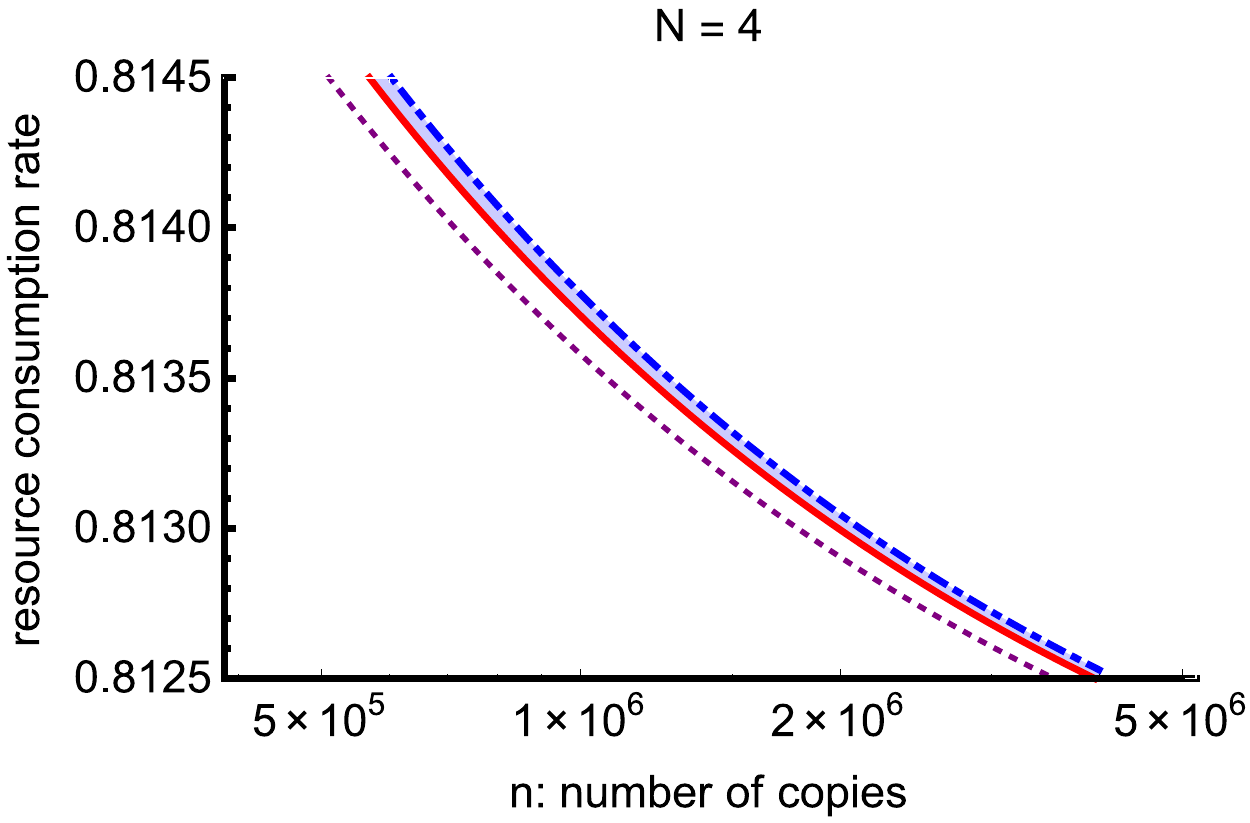}
\caption{Comparison of resource consumption rates at edge $e_1=\{v_1,v_2\}$ up to the second order for approximate construction of the $W$ state under the line-topology graph for $N=4$.
        The blue dotted-dashed curve represents the resource consumption rate $R_{\textup{const}}$ when error thresholds at edges are set constantly, and the red solid curve the rate $R_{\textup{opt}}$ when error thresholds are optimized in order to save the total resource consumption.
        The purple dotted curve represents their lower bound $\underline{R}$ up to the second order.
        Optimization of error thresholds at edges provides an efficient algorithm for approximate state construction.
}
\label{fig:graph4}
\end{figure}

\begin{example}[Improvement of approximate state construction]
  Consider approximate construction of $n$ copies of the four-qubit $W$ state over four parties under the same line-topology graph as the one in Example~\ref{ex:4} (i.e.,\ $N=4$).  In Example~\ref{ex:4}, we took
    \begin{equation}
        \epsilon'(e_1)=\epsilon'(e_2)=\epsilon'(e_3)=3^{-1/2}\epsilon.
        \label{eq:const_error}
    \end{equation}
    However, approximation over the edge $e_2=\{v_2,v_3\}$ does not contribute to saving the total resource consumption, since $b\left(\rho_{e_2},\epsilon'(e_2)\right)$ is zero, i.e.,\ the second term of the second-order expansion of the approximate graph-associated entanglement cost is zero, in this case.
Therefore, by taking
\begin{equation}
    \begin{split}
        \epsilon'\left(e_2\right)&=0,\\
        \epsilon'\left(e_1\right)=\epsilon'\left(e_3\right)&=2^{-1/2}\epsilon,
    \end{split}
    \label{eq:optimized_error}
\end{equation}
we can save the total resource consumption.
Figure~\ref{fig:graph4} illustrates the comparison of the resource consumption rates at $e_1=\{v_1,v_2\}$.
We compare the rate $R_{\textup{const}}$ when using Eq.~\eqref{eq:const_error} as error thresholds and the rate $R_{\textup{opt}}$ when using Eq.~\eqref{eq:optimized_error}, up to the second order, i.e.,\
\begin{equation*}
    \begin{split}
        R_{\textup{const}}(n)&=a\left(\rho_{e_1}\right)+b\left(\rho_{e_1},3^{-1/2}\epsilon\right)\frac{1}{\sqrt{n}},\\
        R_{\textup{opt}}(n)&=a\left(\rho_{e_1}\right)+b\left(\rho_{e_1},2^{-1/2}\epsilon\right)\frac{1}{\sqrt{n}}.
    \end{split}
\end{equation*}
We write their lower bound up to second order as
\[
    \underline{R}(n)=a\left(\rho_{e_1}\right)+b\left(\rho_{e_1},\epsilon\right)\frac{1}{\sqrt{n}}.
\]
While all $R_{\textup{const}}$, $R_{\textup{opt}}$ and $\underline{R}$ converge to the same asymptotic rate $a\left(\rho_{e_1}\right)=0.811\cdots$ as $n\to \infty$,
the second-order asymptotic analysis provides an efficient strategy through optimizing error thresholds for finite $n$.

\end{example}

\section{\label{sec:5}Conclusion}
In this paper, we introduced and analyzed the graph-associated entanglement costs of the exact and approximate state construction tasks for operational characterizations of multipartite entanglement.   The graph-associated entanglement costs represent the optimal resource state consumed at edges of a given graph in the state construction tasks.

We evaluated the edge graph-associated entanglement costs under any tree by presenting efficient distributed algorithms for construction of multipartite pure states.
Our distributed algorithms for state construction can reduce required resources compared to a na\"{\i}ve distribution method of the state in which the root party prepares the whole target state in the beginning and then distributes a corresponding part of the state to each party.  This is because our algorithms construct the target state by ``extending'' the state part by part at each party for the next forward neighboring parties.

The graph-associated entanglement costs incorporate the topology of graphs to characterize multipartite entanglement, in addition to the amount of bipartite entanglement resources at each edge of the graphs.  We found that there is a class of multipartite states where the exact total graph-associated entanglement cost only depends on the number of the parties, i.e., \textit{independent} of the graph topology.  In this class are the $W$ states and GHZ states.  Further topology-based classification should be possible, which we leave as future works.

There are several more issues we would like to point out for future works.   The lower bound of the resource consumption in approximate construction of finite copies of a target state within a fixed error threshold should be further improved.
To improve the lower bound, the analysis of multi-party LOCC may gain importance, as well as that of the combinatorial optimization to present the upper bound explicitly.
The optimal state construction algorithms of multipartite states under an arbitrary graph also remain.

\begin{acknowledgments}
This work is supported by the Project for Developing Innovation Systems of MEXT, Japan, and JSPS
KAKENHI (Grants No.~26330006, No.~15H01677, No.~16H01050, and No.~17H01694). We also acknowledge the ELC project [Grant-in-Aid for Scientific Research on Innovative Areas MEXT KAKENHI (Grant No.~24106009)].
\end{acknowledgments}

\appendix

\section{\label{app:a}Proofs for exact state construction}

We present a proof of Theorem~\ref{thm:exact} by showing that the upper and lower bounds of $E_{\textup{GC},i,e}^T$ coincide.
The upper bound is proved by explicitly presenting a distributed algorithm for the exact state construction, while the lower bound is derived from the analysis for the bipartition with respect to each edge.

\subsection{Proof of the lower bound of Theorem~\ref{thm:exact}} 
Consider a bipartition of the tree $T=(V,E)$ with respect to an edge $e=\{p(v),v\}\in E$ induced by removing an edge $e$, namely, the decomposition of the set $V$ into two subsets $D'_v$ and $\overline{D'_v}$.   Since the resource state for $T$ is given by $\Ket{\Phi_{\textup{res}}(T)}=\bigotimes_{e\in E}\Ket{\Phi_{m_e}^+}^e$ where $m_e$ represents the Schmidt rank of the maximally entangled state corresponding to $e$, the Schmidt rank of $\Ket{\Phi_{\textup{res}}(T)}$ for the bipartition with respect to $e$ is $m_e$.
As the Schmidt rank does not increase under LOCC transformations, the Schmidt rank of the target state $\Ket{\psi}$ for the bipartition with respect to $e$, represented by $R_e$, has to satisfy
\begin{equation}
    R_e \leqq m_e,
\label{eq:necessity_oneshot}
\end{equation}
if $\ket{\psi}$ is transformed from $\Ket{\Phi_{\textup{res}}(T)}$ via LOCC\@.  The inequality~\eqref{eq:necessity_oneshot} holds for any $m_{e}$.  Therefore we obtain
\begin{equation}
    \label{eq:exact_lower}
    E_{\textup{GC},i,e}^T\left(\psi\right)\geqq\log_2 R_e.
\end{equation}

\subsection{Recursive description of multipartite states}

To derive the upper bound, we use the recursive description of multipartite states presented in Proposition~\ref{lem:decomposition_tree}.
\begin{proof}[Proof of Proposition~\ref{lem:decomposition_tree}]
    We prove Eq.~\eqref{eq:w_decomposition}, while Eq.~\eqref{eq:psi_decomposition} can be proved in a similar way.
    Fix $v\in V$ which is not the root or a leaf.   Consider the Schmidt decomposition of the target state $\Ket{\psi}\in\mathcal{H}_t$ with respect to the bipartition induced by removing edge $e=\{p(v),v\}$ given by 
\[
\Ket\psi = \sum_{l_v=0}^{R_e-1}\sqrt{\lambda_{l_v}}\Ket{w_{l_v}}\otimes\Ket{\overline{w}_{l_v}},
\] 
corresponding to Eq.~\eqref{eq:schmidt_decomposition}. Each Schmidt basis state $\Ket{w_{l_v}} \in \mathcal{H}^{D'_v}_t$ can be expanded as
\begin{equation}
\label{eq:w_to_l_and_q}
    \Ket{w_{l_v}} = \sum_{l}\alpha_{l,l_v}^{v}\Ket{l}\otimes\Ket{q_{l, l_v}},
\end{equation}
where ${\left\{\Ket{l}\right\}}_l$ is the computational basis of $\mathcal{H}^v_t$ and ${\{ \Ket{q_{l, l_v}}\}}_k$ is a set of normalized but not necessary orthogonal states of $\mathcal{H}^{D_v}_t$.   Thus it is sufficient to prove that $\Ket{q_{l, l_v}}$ can be written as
    \begin{equation}
        \Ket{q_{l, l_v}} =
        \sum_{\boldsymbol{l}_{C_v}}
        \beta_{l,\boldsymbol{l}_{C_v},l_v}^{v}  \bigotimes_{c\in C_v}\Ket{w_{l_c}}
        \label{eq:q_l_m_v_k}
    \end{equation}
    with the Schmidt basis state $\Ket{w_{l_c}}$ of $\mathcal{H}^{D'_c}_t$ for $c \in C_v$.

    Consider, for each $c\in C_v$, any $\Ket{w_c^{\perp}}\in\mathcal{H}^{D'_c}_t$ such that
\[
    \forall l_c\in\left\{0,1,\ldots,R_{\left\{v,c\right\}}-1\right\}, \langle w_c^{\perp}|w_{l_c}\rangle = 0.
\]
    It holds that
    \[
        \left(\Bra{w_c^{\perp}}\otimes \openone\right)\Ket{\psi} = 0,
    \]
    since otherwise it contradicts $\langle w_c^{\perp}|\left(\tr_{\overline{D'_c}}\psi\right)|w_c^{\perp}\rangle=0$.
    It is sufficient to prove that $\Ket{q_{l, l_v}}$ can be chosen so that, for any $\Ket{w_c^{\perp}}$, it holds that
    \[
        \left(\Bra{w_c^{\perp}}\otimes \openone\right)\Ket{q_{l, l_v}} = 0.
    \]

A straightforward calculation shows that
    \[
        \begin{split}
            0 =& \left(\Bra{w_c^{\perp}}\otimes \openone\right)\Ket{\psi} \\
            =& \sum_{l_v}\sqrt{\lambda_{l_v}}\left(\sum_{l}\alpha^{v}_{l,l_v}\Ket{l}\otimes 
                \left(\Bra{w_c^{\perp}}\otimes \openone\right)\Ket{q_{l, l_v}}\right)\otimes
            \Ket{\overline{w_{l_v}}}.
        \end{split}
    \]
    Since ${\{ \Ket{l}\otimes\Ket{\overline{w}_{l_v}} \}}_{l, l_v}$ is a set of mutually linearly independent states and $\sqrt{\lambda_{l_v}}>0$, it holds that $\alpha^{v}_{l,l_v}=0$ or 
\begin{equation}
        \left(\Bra{w_c^{\perp}}\otimes \openone\right)\Ket{q_{l, l_v}}=0.
\label{eqn:orthogonal}
\end{equation}
For the case $\alpha^{v}_{l,l_v}=0$, the corresponding $\Ket{q_{l, l_v}}$ can be chosen arbitrarily since $\Ket{\psi} = \sum_{l,l_v}\sqrt{\lambda_{l_v}}\alpha^{v}_{l,l_v}\Ket{l}\otimes \Ket{q_{l, l_v}}\otimes
\Ket{\overline{w_{l_v}}}$.  Therefore, in any case, $\Ket{q_{l, l_v}}$ can be chosen to satisfy Eq.~\eqref{eqn:orthogonal}.

    In conclusion, for any $v\in V$ which is not the root or a leaf, we obtain
    \begin{align*}
        &\Ket{w_{l_v}} = \sum_{l}\alpha^{v}_{l,l_v} \Ket{l}  \otimes 
        \sum_{\boldsymbol{l}_{C_v}}
        \beta_{l,\boldsymbol{l}_{C_v},l_v}^{v}  \bigotimes_{c\in C_{v}}\Ket{w_{l_c}}.
    \end{align*}
    To show~Eq.~\eqref{eq:psi_decomposition}, replace $v_1$ by $v$, remove $l_v$ in Eqs.~\eqref{eq:w_to_l_and_q} and~\eqref{eq:q_l_m_v_k}, and follow the above argument for the proof to obtain
    \begin{align*}
    \Ket{\psi} = \sum_{l}\alpha^{v_1}_l \Ket{l}  \otimes 
    \sum_{\boldsymbol{l}_{C_{v_1}}}
    \beta_{l,\boldsymbol{l}_{C_{v_1}}}^{v_1}  \bigotimes_{c\in C_{v_1}}\Ket{w_{l_c}}.
 \end{align*}

\end{proof}

\subsection{Completeness of measurements}
    In the distributed algorithm for exact state construction based on Proposition~\ref{lem:decomposition_tree}, the measurements represented by Eqs.~\eqref{eq:POVM_root} and~\eqref{eq:POVM_each} are used.
    We show the completeness of the measurements.
    \begin{proposition}[Completeness of the measurements for exact state construction]
        For all nonleaf $v\in V$,
\begin{equation*}
            \sum_{\boldsymbol{x}_{C_v},\boldsymbol{z}_{C_v}}M^{v\dag}_{\boldsymbol{x}_{C_v},\boldsymbol{z}_{C_v}} M^{v}_{\boldsymbol{x}_{C_v},\boldsymbol{z}_{C_v}}=\textup{\openone}.
\end{equation*}
\end{proposition}
    \begin{proof}
        We prove the case for nonroot  (and nonleaf) $v$, while the case for the root can be proved by removing indices $l_v$ and $l'_v$ in the following argument.

        For generalized Pauli operators $X(x)$ and $Z(z)$ of a $d$-dimensional Hilbert space $\mathcal{H}_d$ defined by Eqs.~\eqref{eq:HWopx} and~\eqref{eq:HWopz},  a channel $\rho \rightarrow (1/d^2)\sum_{x, z=0}^{d-1}{X(x)}^\dag {Z(z)}^\dag\rho Z(z)X(x)$ for any state $\rho$ on $\mathcal{H}_d$ is a completely depolarizing channel, namely, $(1/d^2)\sum_{x, z=0}^{d-1} {X(x)}^\dag {Z(z)}^\dag\rho Z(z)X(x)= \openone /d$.
Then we have
        \begin{equation*}
            \begin{split}
                &\sum_{\boldsymbol{x}_{C_v},\boldsymbol{z}_{C_v}}M^{v\dag}_{\boldsymbol{x}_{C_v},\boldsymbol{z}_{C_v}} M^{v}_{\boldsymbol{x}_{C_v},\boldsymbol{z}_{C_v}}=\\
                &\sum_{l_v,l'_v}\sum_{l}\overline{\alpha^{v}_{l,l'_v}}\alpha^{v}_{l,l_v}\Ket{l'_v}\Bra{l_v}\otimes\\
                &\sum_{\boldsymbol{l}_{C_v},\boldsymbol{l'}_{C_v}}\overline{\beta_{l,\boldsymbol{l'}_{C_v},l'_v}^{v}}\beta_{l,\boldsymbol{l}_{C_v},l_v}^{v}
    \bigotimes_{c\in C_v}\openone.
\end{split}
\end{equation*}
On the other hand, orthogonality of the Schmidt basis $\left\{\Ket{w_{l_v}}\right\} _{l_v}$ yields
\begin{equation*}
            \begin{split}
                \sum_{l}\sum_{\boldsymbol{l}_{C_v},\boldsymbol{l'}_{C_v}}\overline{\alpha^{v}_{l,l'_v}}\alpha^{v}_{l,l_v}
                \overline{\beta_{l,\boldsymbol{l'}_{C_v},l'_v}^{v}}\beta_{l,\boldsymbol{l}_{C_v},l_v}^{v} &= \langle w_{l'_v} | w_{l_v}\rangle\\
                &= \delta_{l'_v,l_v}.
\end{split}
    \end{equation*}
    Therefore, we obtain
        \begin{equation*}
            \sum_{\boldsymbol{x}_{C_v},\boldsymbol{z}_{C_v}}M^{v\dag}_{\boldsymbol{x}_{C_v},\boldsymbol{z}_{C_v}} M^{v}_{\boldsymbol{x}_{C_v},\boldsymbol{z}_{C_v}}=\sum_{l_v}\Ket{l_v}\Bra{l_v}\otimes\bigotimes_{c\in C_{v}}\openone=\openone.
\end{equation*}
    \end{proof}

\subsection{Distributed algorithm for exact state construction}
We show a distributed algorithm consisting of three sub-algorithms (Algorithm~\ref{alg:1}, Algorithm~\ref{alg:2} and Algorithm~\ref{alg:3}) presented in Table~I for constructing any pure target state $\Ket{\psi}$ from the optimal initial resource state $\ket{\Phi_{\textup{res}}(T)}=\bigotimes_{e\in E}\Ket{\Phi^+_{R_e}}^e$ for tree $T=(V,E)$, which yields the upper bound of $E^T_{\textup{GC},i,e}$ of Theorem~\ref{thm:exact}.

\begin{table}[H]
\label{algorithmtable}
\caption{Three sub-algorithms consisting of the distributed algorithm for exact construction of a target state $\Ket{\psi}$ from the optimal initial resource state $\ket{\Phi_{\textup{res}}(T)}=\bigotimes_{e\in E}\Ket{\Phi^+_{R_e}}^e$ for tree $T=(V,E)$.}
\begin{algorithm}[H]
    \caption{\label{alg:1}The algorithm $\mathcal{A}^{v_1}$ for the root $v_1\in V$}
    \begin{algorithmic}[1]
        \ForAll{$c\in C_{v_1}$}
        \State{Calculate the Schmidt basis $\left\{\Ket{w_{l_c}}\right\} _{l_c}$ of $\mathcal{H}_t^{D'_c}$ in \par Eq.~\eqref{eq:psi_decomposition}.}
        \State{Fix a computational basis $\left\{\Ket{l_c}\right\} _{l_c}$ of $\mathcal{H}^{v_1}_{r_{\left\{v_1,c\right\}}}$ in {\par}Eq.~\eqref{eq:POVM_root}.}
        \State{Send to $c$ the classical description of the above bases.}
    \EndFor{}
    \State{Fix a computational basis $\left\{\Ket{l}\right\} _{l}$ of $\mathcal{H}_{t}^{v_1}$.}
    \State{Calculate $M^{v_1}_{\boldsymbol{x}_{C_{v_1}},\boldsymbol{z}_{C_{v_1}}}$ in Eq.~\eqref{eq:POVM_root}.}
    \State{Perform on $\mathcal{H}^{v_1}_r$ a measurement according to Eq.~\eqref{eq:POVM_root}.}
    \Comment{This consumes the initial resource state shared between $v_1$ and each child $c\in C_{v_1}$.}
    \ForAll{$c\in C_{v_1}$}
    \State{Send to $c$ the measurement outcome $\left(x_c,z_c\right)$.}
\EndFor{}
\end{algorithmic}
\end{algorithm}
\begin{algorithm}[H]
    \caption{\label{alg:2}The algorithm $\mathcal{A}^{v}$ for any $v\in V$ which is not the root or a leaf.}
\begin{algorithmic}[1]
    \State{Receive from the parent $p(v)$ the measurement outcome $\left(x_v,z_v\right)$ and the description of $\left\{\Ket{w_{l_v}}\right\} _{l_v}$ and $\left\{\Ket{l_v}\right\} _{l_v}$.}
    \ForAll{$c\in C_v$}
    \State{Calculate the Schmidt basis $\left\{\Ket{w_{l_c}}\right\} _{l_c}$ of $\mathcal{H}^{D'_c}_{t}$ in \par Eq.~\eqref{eq:w_decomposition}.}
    \State{Fix a computational basis $\left\{\Ket{l_c}\right\} _{l_c}$ of $\mathcal{H}^v_{r_{\{v,c\}}}$ in \par Eq.~\eqref{eq:POVM_each}.}
    \State{Send to $c$ the classical description of the above bases.}
    \EndFor{}
    \State{Fix a computational basis $\left\{\Ket{l}\right\} _{l}$ of $\mathcal{H}_{t}^{v}$.}
    \State{Calculate $M^{v}_{\boldsymbol{x}_{C_v},\boldsymbol{z}_{C_v}}$ in Eq.~\eqref{eq:POVM_each}.}
    \State{Apply the correction $Z\left(-z_v\right)X\left(x_v\right)$ on $\mathcal{H}^v_{r_{\{p(v),v\}}}$.}
    \State{Perform on $\mathcal{H}^{v}_r$ a measurement according to Eq.~\eqref{eq:POVM_each}.}
    \Comment{This consumes the initial resource state shared between $v$ and each child $c\in C_v$.}
    \ForAll{$c\in C_v$}
    \State{Send to $c$ the measurement outcome $\left(x_c,z_c\right)$.}
\EndFor{}
\end{algorithmic}
\end{algorithm}
\begin{algorithm}[H]
    \caption{\label{alg:3}The algorithm $\mathcal{A}^{v}$ for any leaf $v\in V$.}
\begin{algorithmic}[1]
    \State{Receive from the parent $p(v)$ the measurement outcome $\left(x_v,z_v\right)$ and the description of $\left\{\Ket{w_{l_v}}\right\} _{l_v}$ and $\left\{\Ket{l_v}\right\} _{l_v}$.}
    \State{Apply the correction $Z\left(-z_v\right)X\left(x_v\right)$ on $\mathcal{H}^v_r$.}
    \State{Perform on $\mathcal{H}_{r}^{v}$ the isometry transformation $U_{v}\colon\mathcal{H}_{r}^{v}\rightarrow\mathcal{H}_{t}^{v}$ which transforms the basis as
        \begin{align*}
            \Ket{l_v}\xrightarrow{U_v}\Ket{w_{l_v}}.
        \end{align*}}
\end{algorithmic}
\end{algorithm}
\end{table}

\begin{proof}[Proof of the upper bound of Theorem~\ref{thm:exact}]
First, a vertex on $T$ is chosen as the root $v_1 \in V$ of $T$.   In case assigning the root is a nontrivial task, employ classical processing based on the leader election protocols~\cite{RefWorks:142}.  Any vertex can be chosen as the root of the tree.
Consider that the party at the root $v_1 \in V$ performs Algorithm~\ref{alg:1}, then the parties at the vertex $v\in V$ which is not the root or a leaf perform Algorithm~\ref{alg:2}, and finally the parties at the leaf $v_{\textup{leaf}}\in V$ performs Algorithm~\ref{alg:3}.
After all the algorithms terminated on all the vertices, the target state $\Ket{\psi}$ has been deterministically constructed and shared among the parties.
The existence of the distributed algorithm yields
\begin{equation}
    \label{eq:exact_upper}
    \sum_{e\in E}E_{\textup{GC},i,e}^T\left(\psi\right)\leqq\sum_{e\in E}\log_2 R_e.
\end{equation}
\end{proof}

Therefore the lower bound of $E_{\textup{GC},i,e}^T$ given by Eq.~\eqref{eq:exact_lower} and the upper bound of $E_{\textup{GC},i,e}^T$  given by Eq.~\eqref{eq:exact_upper} are shown to coincide and we have proven that
\[
    E_{\textup{GC},i,e}^T\left(\psi\right)=\log_2 R_e,
\]
where $i$ is uniquely determined.

\section{\label{app:demo} Measurement operators for construction of matrix product states}

We show an explicit description of the measurement operators for constructing a target state given by the canonical form of matrix product states (MPS).  We consider construction of an MPS which is an $N$-partite state $\Ket\psi$ under a line-topology graph $G=(V,E)$, where we designate a vertex at an end of the line topology as the root $v_1\in V$ and label vertices and edges as $V=\left\{v_1,v_2,\ldots,v_N\right\}$, $E=\left\{e_k = \{v_k,v_{k+1}\}\colon k\in\{1,2,\ldots,N-1\}\right\}$.  The order of parties performing measurement in our algorithm is determined as $v_1,v_2,\ldots,v_N$.

The canonical form of MPS describing $\Ket{\psi}$ is given by
\begin{equation}
    \begin{split}
        \Ket\psi =& \sum_{i_1,i_2,\ldots,i_N}\sum_{\alpha_1,\alpha_2,\ldots,\alpha_{N-1}}
        \left(\Gamma_{\alpha_1}^{[1]i_1}\lambda_{\alpha_1}^{[1]}\Gamma_{\alpha_1 \alpha_2}^{[2]i_2}\lambda_{\alpha_2}^{[2]}\cdots\right.\\
        &\left.\lambda_{\alpha_{N-1}}^{[N-1]}\Gamma_{\alpha_{N-1}}^{[N]i_N}\right)\Ket{i_1}^{1}\otimes\Ket{i_2}^{2}\otimes\cdots\otimes\Ket{i_N}^{N},
    \end{split}
    \label{eq:canonical}
\end{equation}
where the sum is over all the possible values of the indices, ${\{\Ket{i_k}\}}_{i_k}$ for each $k\in\{1,2,\ldots,N\}$ is a fixed computational basis of $\mathcal{H}^{v_k}_t$, $\Gamma^{[1]}$ is a rank-$2$ tensor, $\lambda^{[1]}$ rank-$1$, $\Gamma^{[2]}$ rank-$3$, and so on.
Given a state $\Ket\psi$, its canonical form can be obtained from the following procedure~\cite{APractical}.
Firstly, consider a Schmidt decomposition of $\Ket\psi$ for the bipartition with respect to $e_1$, i.e.\
\[
\Ket\psi = \sum_{\alpha_1}\sqrt{\lambda_{\alpha_1}}\Ket{w_{\alpha_1}}\otimes\Ket{\overline{w_{\alpha_1}}},
\]
where $\sqrt{\lambda_{\alpha_1}}$ for each $\alpha_1$ is the Schmidt coefficient, and ${\left\{\Ket{w_{\alpha_1}}\right\}}_{\alpha_1}$ and ${\left\{\Ket{\overline{w_{\alpha_1}}}\right\}}_{\alpha_1}$ are the Schmidt basis of $\mathcal{H}^{D'_{v_2}}_t$ and $\mathcal{H}^{v_1}_t$ respectively.
Here, choose $\Gamma^{[1]}$ and $\lambda^{[1]}$ so that, for all $\alpha_1$, they satisfy
\begin{align*}
    \Ket{\overline{w_{\alpha_1}}} &= \sum_{i_1}\Gamma_{\alpha_1}^{[1]i_1}\Ket{i_1},\\
    \sqrt{\lambda_{\alpha_1}} &= \lambda_{\alpha_1}^{[1]}.
\end{align*}
The target state $\Ket\psi$ is written as
\[
    \Ket\psi = \sum_{i_1}\sum_{\alpha_1}\Gamma_{\alpha_1}^{[1]i_1}\lambda_{\alpha_1}^{[1]}\Ket{i_1}^{1}\otimes\Ket{w_{\alpha_1}}.
\]
Secondly, consider another Schmidt decomposition of $\Ket\psi$ for the bipartition with respect to $e_2$, i.e.\
\[
\Ket\psi = \sum_{\alpha_2}\sqrt{\lambda_{\alpha_2}}\Ket{w_{\alpha_2}}\otimes\Ket{\overline{w_{\alpha_2}}},
\]
where $\sqrt{\lambda_{\alpha_2}}$ is the Schmidt coefficient, and ${\left\{\Ket{w_{\alpha_2}}\right\}}_{\alpha_2}$ and ${\left\{\Ket{\overline{w_{\alpha_2}}}\right\}}_{\alpha_2}$ are the Schmidt basis of $\mathcal{H}^{D'_{v_3}}_t$ and $\mathcal{H}^{\overline{D'_{v_3}}}_t$ respectively.
Then, calculate $\Gamma^{[2]}$ and $\lambda^{[2]}$ so that, for all $\alpha_1, \alpha_2$, they satisfy the following decomposition of the Schmidt basis $\Ket{w_{\alpha_1}}$ of $\mathcal{H}^{D'_{v_2}}_t=\mathcal{H}^{v_2}_t\otimes\mathcal{H}^{D'_{v_3}}_t$,
\begin{align*}
    \Ket{w_{\alpha_1}} &= \sum_{i_2}\sum_{\alpha_2}\Gamma_{\alpha_1 \alpha_2}^{[2]i_2}\lambda_{\alpha_2}^{[2]}\Ket{i_2}\otimes\Ket{w_{\alpha_2}},\\
    \sqrt{\lambda_{\alpha_2}} &= \lambda_{\alpha_2}^{[2]}.
\end{align*}
Therefore, $\Ket\psi$ is written as
\[
    \Ket\psi = \sum_{i_1,i_2}\sum_{\alpha_1,\alpha_2}\Gamma_{\alpha_1}^{[1]i_1}\lambda_{\alpha_1}^{[1]}\Gamma_{\alpha_1 \alpha_2}^{[2]i_2}\lambda_{\alpha_2}^{[2]}\Ket{i_1}^{1}\otimes\Ket{i_2}^{2}\otimes\Ket{w_{\alpha_2}}.
\]
Iterating recursively the above yields
\begin{equation*}
    \begin{split}
        \Ket\psi =& \sum_{i_1,i_2,\ldots,i_N}\sum_{\alpha_1,\alpha_2,\ldots,\alpha_{N-1}}
        \left(\Gamma_{\alpha_1}^{[1]i_1}\lambda_{\alpha_1}^{[1]}\Gamma_{\alpha_1 \alpha_2}^{[2]i_2}\lambda_{\alpha_2}^{[2]}\cdots\right.\\
        &\left.\lambda_{\alpha_{N-1}}^{[N-1]}\right)\Ket{i_1}^{1}\otimes\cdots\otimes\Ket{i_{N-1}}^{N-1}\otimes\Ket{w_{\alpha_{N-1}}},
    \end{split}
\end{equation*}
where ${\left\{\Ket{w_{\alpha_{N-1}}}\right\}}_{\alpha_{N-1}}$ is the Schmidt basis of $\mathcal{H}^{v_N}_t$.
Finally, take $\Gamma^{[N]}$ so that it satisfies
\[
    \Ket{w_{\alpha_{N-1}}}=\sum_{i_N}\Gamma_{\alpha_{N-1}}^{[N]i_N}\Ket{i_N},
\]
to obtain Eq.~\eqref{eq:canonical}.

Comparing with the notation in Proposition~\ref{lem:decomposition_tree}, we obtain
\begin{align*}
    \alpha^{v_1}_l\beta_{l,l_{v_2}}^{v_1}&=\Gamma_{l_{v_2}}^{[1]l}\lambda_{l_{v_2}}^{[1]},\\
    \alpha^{v_k}_{l,l_{v_k}}\beta_{l,l_{v_{k+1}},l_{v_k}}^{v_k}&=\Gamma_{l_{v_k}l_{v_{k+1}}}^{[k]l}\lambda_{l_{v_{k+1}}}^{[k]},\\
\end{align*}
for $k=2,3,\ldots,N-1$.
Therefore, the measurements represented by Eqs.~\eqref{eq:POVM_root} and~\eqref{eq:POVM_each} are written as
\begin{align*}
    M^{v_1}_{x_{v_2},z_{v_2}}=& \sum_{l,l_{v_2}}\Gamma_{l_{v_2}}^{[1]l}\lambda_{l_{v_2}}^{[1]}\\
                              & \Ket{l}\left(\Bra{l_{v_2}} Z(z_{v_2}) X(x_{v_2})/\sqrt{R_{\{v_1,v_2\}}}\right),\\
    M^{v_k}_{x_{v_{k+1}},z_{v_{k+1}}}=& \sum_{l,l_{v_k},l_{v_{k+1}}}\Gamma_{l_{v_k}l_{v_{k+1}}}^{[k]l}\lambda_{l_{v_{k+1}}}^{[k]} \Ket{l}
    \Bra{l_{v_k}} \otimes\\
    &\left(\Bra{l_{v_{k+1}}}Z(z_{v_{k+1}}) X(x_{v_{k+1}})/\sqrt{R_{\{v_k,v_{k+1}\}}}\right),
\end{align*}
for $k=2,3,\ldots,N-1$.
The isometry transformation performed by the leaf $v_N$ is
\[
    U_{v_N}=\sum_{l,l_{v_N}}\Gamma_{l_{v_N}}^{[N]l}\Ket{l}\Bra{l_{v_N}},
\]
as the right-hand side describes the change of basis from $\Ket{l_{v_N}}\in\mathcal{H}^{v_N}_r$ into the Schmidt basis $\Ket{w_{l_{v_N}}}=\sum_{l}\Gamma_{l_{v_N}}^{[N]l}\Ket{l}\in\mathcal{H}^{v_N}_t$.

\section{\label{app:b}Proofs for approximate state construction}
To prove the upper bound given in Theorem~\ref{thm:asymptotic}, we first prove Proposition~\ref{prp:approximation}.  Then the proofs of the upper bound and the lower bound in Theorem~\ref{thm:asymptotic} are presented followed by the proof of Corollary~\ref{cor:asymptotic_cost}.

\subsection{Proof of Proposition 10}
Proposition~\ref{prp:approximation} is proven by using the noncommutative union bound for sequential projections presented 
in Ref.~\cite{pra052331} stated as the following Lemma.
\begin{lemma}[Noncommutative union bound for sequential projections]
\label{lem:union_bound}
    Given any density operator $\rho$ and projections $\Pi_1,\Pi_2,\ldots,\Pi_{N-1}$ satisfying, for each $i\in\{1,2,\ldots,N-1\}$,
    \[
        \tr\Pi_i\rho=1-\epsilon_i,
    \]
    it holds that
    \[
        \left\|\rho-\frac{\Pi_{N-1}\cdots\Pi_2\Pi_1\rho\Pi_1\Pi_2\cdots\Pi_{N-1}}{\tr\Pi_{N-1}\cdots\Pi_2\Pi_1\rho\Pi_1\Pi_2\cdots\Pi_{N-1}}\right\|_1\leqq 2\sqrt{\sum_i \epsilon_i}.
    \]
\end{lemma}

\begin{proof}[Proof of Proposition~\ref{prp:approximation}.]
    Setting $\rho=\psi^{\otimes n}$, $\Pi_i = \Pi_{e_i,\psi}$ for each $i=1,2,\ldots,N-1$ in Lemma~\ref{lem:union_bound}, we obtain
    \[
        \left\|\psi^{\otimes n}-\tilde\psi_n\right\|_1\leqq 2\sqrt{\sum_{e\in E}\left(1-\tr \Pi_{e,\psi}\psi^{\otimes n}\right)}.
    \]

    Since the relation Eq.~\eqref{eq:projection} implies
    \[
        1-\tr \Pi_{e,\psi}\psi^{\otimes n}\leqq \frac{{\epsilon'(e)}^2}{4}
    \]
    for each $e\in E$, we have
    \[
        \left\|\psi^{\otimes n}-\tilde\psi_n\right\|_1\leqq\sqrt{\sum_{e\in E}{\epsilon'(e)}^2}.
    \]
\end{proof}

\subsection{The proof of the upper bound in Theorem~\ref{thm:asymptotic}}

Exact construction of the approximate state provides the upper bound of Theorem~\ref{thm:asymptotic}.
\begin{proof}[Proof of the first statement in Theorem~\ref{thm:asymptotic}.]
For any $\epsilon'$ satisfying
\begin{equation}
    \sqrt{\sum_{e\in E}{\epsilon'(e)}^2}\leqq\epsilon,
\end{equation}
Proposition~\ref{prp:approximation} implies that $\left\|\psi^{\otimes n}-\tilde\psi_n\right\| _1\leqq\epsilon$.
Then, for any $i$, it holds that
\begin{equation}
\label{eq:approximate_exact}
    \sum_{e\in E}E_{\textup{GC},i,e}^{n,\epsilon,T}\left(\psi\right) \leqq \sum_{e\in E}\frac{E_{\textup{GC},i',e}^T\left(\tilde\psi_n\right)}{n},
\end{equation}
where $i'$ in the right hand side is uniquely determined.

$E_{\textup{GC},i',e}^T\left(\tilde\psi_n\right)$ for each $e=\{p(v),v\} \in E$ can be evaluated by using Theorem~\ref{thm:exact} as follows.  From Theorem~\ref{thm:exact}, we have
\[
    \begin{split}
        &E_{\textup{GC},i,e}^T\left(\tilde\psi_n\right)\\
        &=\log_2\rank\tr_{\overline{D'_v}}\tilde\psi_n\\
        &=\log_2 \rank\\
        &\quad\tr_{\overline{D'_v}}\frac{\Pi_{e_{N-1},\psi}\cdots\Pi_{e_1,\psi}\psi^{\otimes n}\Pi_{e_1,\psi}\cdots\Pi_{e_{N-1},\psi}}{\tr\Pi_{e_{N-1},\psi}\cdots\Pi_{e_1,\psi}\psi^{\otimes n}\Pi_{e_1,\psi}\cdots\Pi_{e_{N-1},\psi}}.\\
    \end{split}
\]
As for the argument of the logarithm,
we have
\[
    \begin{split}
        &\rank\tr_{\overline{D'_v}}\frac{\Pi_{e_{N-1},\psi}\cdots\Pi_{e_1,\psi}\psi^{\otimes n}\Pi_{e_1,\psi}\cdots\Pi_{e_{N-1},\psi}}{\tr\Pi_{e_{N-1},\psi}\cdots\Pi_{e_1,\psi}\psi^{\otimes n}\Pi_{e_1,\psi}\cdots\Pi_{e_{N-1},\psi}}\\
        &=\rank\tr_{\overline{D'_v}}\Pi_{e_{N-1},\psi}\cdots\Pi_{e_1,\psi}\psi^{\otimes n}\Pi_{e_1,\psi}\cdots\Pi_{e_{N-1},\psi}\\
        &\leqq\rank\tr_{\overline{D'_v}}\Pi_{e,\psi}\\
        &=\rank\tr_{\overline{D'_v}}\left(\Pi_{\rho_e}\otimes\openone\right)\\
        &= \rank \Pi_{\rho_{e}}.
    \end{split}
\]
As $\Pi_{\rho_e}$ is the projector onto the subspace spanned by the eigenvectors corresponding to the eigenvalues of $\rho_e^{\otimes n}$ no less than $2^{-\gamma_e}$ where $\gamma_e=\overline{H}_{\textup{S}}^{{\epsilon'(e)}^2/4}\left(\rho_e^{\otimes n}\right)$, the number of the eigenvectors is no more than $2^{\gamma_e}$, namely,
\[
    \rank\Pi_{\rho_e}\leqq 2^{\gamma_e}.
\]
Therefore, it holds that
\begin{equation}
    \label{eq:evaluate_E1_approximate}
    \begin{split}
        E_{\textup{GC},i',e}^T\left(\tilde\psi_n\right)&\leqq\gamma_e\\
        &=\overline{H}_{\textup{S}}^{{\epsilon'(e)}^2/4}\left(\rho_e^{\otimes n}\right).
\end{split}
\end{equation}

Combining Eqs.~\eqref{eq:approximate_exact} and~\eqref{eq:evaluate_E1_approximate}, we obtain
\[
    \sum_{e\in E}E_{\textup{GC},i,e}^{n,\epsilon,T}\left(\psi\right) \leqq \sum_{e\in E}\frac{\overline{H}_{\textup{S}}^{{\epsilon'(e)}^2/4}\left(\rho_e^{\otimes n}\right)}{n}.
\]
\end{proof}

\subsection{Proof of the lower bound in Theorem~\ref{thm:asymptotic}}

\begin{proof}[Proof of the second statement in Theorem~\ref{thm:asymptotic}]
Consider a bipartition of the tree $T=(V,E)$ with respect to $e=\{p(v),v\} \in E$.  We represent the $(n,\epsilon)$-approximate \textit{bipartite} entanglement cost of $\Ket\psi\in\mathcal{H}_t$ for the bipartition with respect to $e$ by $E_{\textup{C},e}^{n,\epsilon,T}{\left(\psi\right)}$.  In the state construction task corresponding to $E_{\textup{C}, e}^{n,\epsilon,T}{\left(\psi\right)}$, quantum communications between the parties within the same side of the partition are allowed but communications across the partitions are limited to classical communications. Since the allowed resources in the state construction task corresponding to $E_{\textup{GC},i,e}^{n,\epsilon,T}{\left(\psi\right)}$ are more restricted than those corresponding to $E_{\textup{C},e}^{n,\epsilon,T}{\left(\psi\right)}$, we have
\begin{equation*}
    E_{\textup{GC},i,e}^{n,\epsilon,T}{\left(\psi\right)}\geqq E_{\textup{C},e}^{n,\epsilon,T}{\left(\psi\right)}.
\end{equation*}

The  $(n,\epsilon)$-approximate bipartite entanglement cost can bounded by the second-order asymptotic analysis presented in Ref.~\cite{RefWorks:160} (Theorem 8 in Ref.~\cite{RefWorks:160}) as
\begin{equation*}
E_{\textup{C},e}^{n,\epsilon,T}{\left(\psi\right)} \geqq \frac{\overline{H}_{\textup{S}}^{\epsilon^2/4+\eta}\left(\rho_e^{\otimes n}\right)-\delta+\log_2 \eta}{n}.
\end{equation*}
Therefore, for any $i$, $e\in E$, $\eta>0$ and $\delta>0$, we obtain
\[
    E_{\textup{GC},i,e}^{n,\epsilon,T}{\left(\psi\right)}\geqq \frac{\overline{H}_{\textup{S}}^{\epsilon^2/4+\eta}\left(\rho_e^{\otimes n}\right)-\delta+\log_2 \eta}{n}.
\]
\end{proof}

\subsection{Proof of Corollary~\ref{cor:asymptotic_cost}}

\begin{proof}[Proof of Corollary~\ref{cor:asymptotic_cost}]
    For any $i$, combining Theorem~\ref{thm:asymptotic} and the second-order expansion of the quantum information spectrum entropy yields
    \begin{align*}
            &\lim_{\epsilon\to 0}\lim_{n\to\infty}\sum_{e\in E}E_{\textup{GC},i,e}^{n,\epsilon,T}\left(\psi\right)\\
            &\leqq \lim_{\epsilon\to 0}\lim_{n\to\infty}\sum_{e\in E}\frac{\overline{H}_{\textup{S}}^{{\epsilon'(e)}^2/4}\left(\rho_e^{\otimes n}\right)}{n}\\
            &= \sum_{e\in E}S\left(\rho_e\right),
    \end{align*}
    and
    \[
        \begin{split}
            &\lim_{\epsilon\to 0}\lim_{n\to\infty}E_{\textup{GC},i,e}^{n,\epsilon,T}\left(\psi\right)\\
            &\geqq \lim_{\epsilon\to 0}\lim_{n\to\infty}\frac{\overline{H}_{\textup{S}}^{\epsilon^2/4+\eta}\left(\rho_e^{\otimes n}\right)-\delta+\log_2 \eta}{n}\\
            &= S\left(\rho_e\right).
        \end{split}
    \]

    Therefore, the approximate edge graph-associated entanglement costs for all $i$ coincide in the asymptotic limit, and we obtain, for each $e\in E$,
    \[
        E_{\textup{GC},i',e}^{\infty,T}\left(\psi\right)=\lim_{\epsilon\to 0}\lim_{n\to\infty}E_{\textup{GC},i,e}^{n,\epsilon,T}\left(\psi\right)=S\left(\rho_e\right),
    \]
    where $i'$ on the left-hand side is uniquely determined.
\end{proof}


\end{document}